\documentclass[11pt,reqno]{amsproc}
\usepackage[margin=1in]{geometry}
\usepackage{amsmath, amsthm, amssymb}
\usepackage[colorlinks=true, pdfstartview=FitV, linkcolor=blue,citecolor=blue, urlcolor=blue]{hyperref}
\usepackage{graphicx}
\usepackage[abbrev,lite,nobysame]{amsrefs}
\usepackage{times}
\usepackage[usenames,dvipsnames]{color}
\usepackage[T1]{fontenc}
\usepackage[utf8]{inputenc} 

\definecolor{labelkey}{rgb}{0,0,1}
\usepackage{dsfont}

\def\R {\mathbb{R}}

\newtheorem{proposition}{Proposition}[section]
\newtheorem{theorem}[proposition]{Theorem}

\newtheorem{lemma}[proposition]{Lemma}
\theoremstyle{definition}
\newtheorem{definition}[proposition]{Definition}
\newtheorem{remark}[proposition]{Remark}

\numberwithin{equation}{section}

\title[Micropolar meets Newtonian in 3D]{Micropolar meets Newtonian in 3D. The Rayleigh--B\'{e}nard problem for large Prandtl numbers.}

\author[P. Kalita]{Piotr Kalita}
\address{Faculty of Mathematics and Computer Science, Jagiellonian University, ul. \L{}ojasiewicza 6, 30-348 Krak\'{o}w, Poland}
\email{piotr.kalita@ii.uj.edu.pl}

\author[G. \L{}ukaszewicz]{Grzegorz \L{}ukaszewicz$^a$\footnote{$^a$Corresponding author. E-mail address: glukasz@mimuw.edu.pl}}
\address{Faculty of Mathematics, Informatics, and Mechanics, University of Warsaw, ul. Banacha 2,  02-097 Warszawa, Poland}
\email{glukasz@mimuw.edu.pl}

\subjclass[2010]{76F35, 76E15,  37L30, 35Q30, 35Q79}

\keywords{Rayleigh--B\'{e}nard problem, Boussinesq system, convective turbulence, micropolar fluid, attractor}
\thanks{This work was supported National Science Center (NCN) of Poland
	under project No. DEC-2017/25/B/ST1/00302. Work of P.K. was also partially supported by NCN of Poland under project No. UMO-2016/22/A/ST1/00077.}
\begin{document}

\begin{abstract}
We consider the Rayleigh--B\'{e}nard problem for the three--dimensional Boussinesq system for the micropolar fluid. We introduce the notion of the multivalued eventual semiflow and prove
the existence of the two-space global attractor $\mathcal{A}^K$ corresponding to weak solutions, for every micropolar parameter $K\geq 0$ denoting the deviation of the considered system from the classical Rayleigh--B\'{e}nard problem for the Newtonian fluid. We prove that for every $K$ the attractor $\mathcal{A}^K$ is the smallest compact, attracting, and invariant set. Moreover, the semiflow restricted to this attractor is single-valued and governed by strong solutions. Further, we prove that the global attractors $\mathcal{A}^K$ converge to 
$\mathcal{A}^0$ upper semicontinuously in Kuratowski sense as $K\to 0$, and that the projection of $\mathcal{A}^0$ on the restricted phase space corresponding to the classical Rayleigh--B\'{e}nard problem is the global attractor for the latter problem, having the invariance property. These results are established under the assumption that the Prandtl number is relatively large with respect to the Rayleigh number.

\begin{center}
\end{center}
\end{abstract}

\maketitle

\section{Introduction and main results}\label{Introduction}
In this paper we consider the Rayleigh--B\'{e}nard problem for the three-dimensional system which models, in the framework of the Boussinesq approximation, the heat convection in  the micropolar fluid. 
The micropolar model is both a simple and significant generalization of the Navier--Stokes model of classical hydrodynamics. It has much more applications than the classical model due to the fact that the latter cannot describe (by definition) fluids with microstructure.
In general, individual particles of such complex fluids (e.g., polymeric 
suspensions, blood, liquid crystals) may be of different shape, may 
shrink and expand, or change their shape, and moreover, they may rotate, 
independently of the rotation and movement of the fluid. To describe accurately the behavior of such fluids one needs a theory that 
takes into account geometry, deformation, and intrinsic motion of individual 
material particles. 
In the framework of continuum mechanics several such theories have appeared, 
e.g., theories of simple microfluids, simple deformable directed fluids, 
micropolar fluids, dipolar fluids, to name some of them.
To account for these local structural aspects, many classical concepts such  
as the symmetry of the stress tensor or absence of couple stresses 
required 
reexamination, and while many principles of classical continuum mechanics 
still remain valid for this new class of fluids, they had to be augmented 
with additional balance laws and constitutive relations. Some of these theories are very general, while others are concerned with 
special 
types of material structure and/or deformation; the potential applicability 
of the various theories is diverse. It is clear that each particular theory has its advantages as well as 
disadvantages when considered from a particular point of view.  
Some theories may seem more sound, logical, justifiable, and useful than 
others. No one of them is universal.

One of the best-established theories of fluids with microstructure is the 
theory of micropolar fluids developed by A.C. Eringen \cite{Eringen66}, and studied from mathematical point of view by \L{}ukaszewicz \cite{L1999}. Physically, micropolar models represent fluids consisting of rigid, 
randomly oriented (or spherical) particles suspended in a viscous medium, 
where the deformation of the particles is ignored. This constitutes a 
substantial generalization of the Navier--Stokes model and opens a new field 
of potential applications including a large number of complex fluids.

Let us point out some general features that make the model  
a~favorite in both theoretical studies and applications. The attractiveness and power of the model of micropolar fluids come from the 
fact that it is both a significant and a simple generalization of the 
classical Navier--Stokes model. 
Only one new vector field, the angular velocity field of rotation of 
particles, denoted by $\gamma$, is introduced. Correspondingly, only one (vector) equation is 
added---it represents the conservation of the angular momentum.
While four new dimensionless constants, denoted by $L, M, G$, and $K$ are introduced, if one of them, namely 
the microrotation viscosity $K$, becomes zero, the conservation law of the 
linear momentum becomes independent of the presence of the microstructure. 
Thus, the magnitude of the microrotation viscosity coefficient $K$  allows us to 
measure, in a certain sense, the deviation of flows of micropolar fluids 
from that of the Navier--Stokes model. 
Thanks to the simplicity of the model of micropolar fluids, in many 
classical applications (e.g., flows through the channel or between parallel 
plates) and under usual geometrical and dynamical assumptions made in such 
cases (e.g., symmetry, linearization of the equations), equations of 
micropolar fluids reduce to ones that can be explicitly (i.e., analytically) solved. 
Thus the solutions obtained, depending on several parameters coming from the 
viscosity coefficients, can be easily compared with solutions of the 
corresponding problems for the Navier--Stokes equations. In addition, as several experiments show, the former solutions better 
represent behavior of numerous real fluids (e.g., blood) than corresponding 
solutions of the classical model, especially when the characteristic 
dimentions of the flow (e.g., the diameter of the channel) become small.
It well agrees with our expectations that the influence of the internal 
structure of the fluid is the greater, the smaller the characteristic 
dimension of the flow.

The simplicity of the micropolar fluid model obviously does not mean 
mathematical triviality. The classical Navier--Stokes model itself, a 
special case of the micropolar fluid model, is far from being trivial. 
In this context simplicity means elegance and beauty of the mathematical 
theory. 

In the nondimensional variables the model system of equations reads \cites{Eringen66, L1999, L2001}
\begin{align}
& \frac{1}{\rm Pr}(u_t + (u\cdot \nabla) u) - (1+K)\Delta u + \frac{1}{\Pr}\nabla p = 
2K {\rm rot}\, \gamma + e_3  {\rm Ra} T& \nonumber\\
& {\rm div}\, u = 0& \nonumber\\
& \frac{M}{\rm Pr}(\gamma_t + (u\cdot \nabla)\gamma) - L\Delta \gamma - G\nabla{\rm div}\,\gamma
   +4K \gamma = 2K {\rm rot}\, u& \nonumber\\
& T_t + u\cdot \nabla T - \Delta T = 0& \nonumber
\end{align}
where ${\rm Pr}$ and ${\rm Ra}$ are the classical Prandtl and Rayleigh numbers, and $L$, $M$, $G$ and $K$ are micropolar parameters, and the unknowns are velocity $u$, microrotation $\gamma$, pressure $p$, and temperature $T$, whose space domain is the set $\Omega\subset \R^3$. The most important parameter is $K=\nu_r/\nu$ which relates the micropolar model to the Navier--Stokes one. The effects of micropolarity depend crucially on the numerical value of parameter $K$, linking the micropolar viscosity $\nu_r$ with the usual kinematic viscosity $\nu$. For $K=0$ there are no effects of micropolarity on the dynamics of the flow, and the above system (the first, second, and fourth equation) reduces to the usual Boussinesq system for classical hydrodynamics.

The main result of the article is the existence of the global attractor $\mathcal{A}^K$ for the studied problem and the convergence, in appropriate upper-semicontinuous sense, of the global attractors $\mathcal{A}^K$  to the attractor $\mathcal{A}^0$ for the Newtonian fluid. As the domain of the considered problem is three dimensional, we cannot rule out the loss of regularity and, generally, it is only known how to prove the existence of the weak attractor. However, if the Prandtl number $\Pr$
is large comparing to other constants present in the model one may use the ideas of Wang \cites{Wang_chapter, Wang2007-Asymptotic,Wang2008}, to get the uniform enstrophy bounds for weak solutions after time dependent on the $L^2$ size of the initial data. We demonstrate that the type of bounds obtained by Wang for Newtonian fluids also holds for micropolar ones. In consequence we prove that although weak solutions may possibly be nonunique, the trajectories can possibly  branch only until the time of regularization is reached. The eventual enstrophy bounds, together with the energy equation method, allow us to get the existence of the two-space global attractor, namely the invariant and compact in $H^1$ set which attracts in $H^1$ the sets of states obtainable from the $L^2$ bounded sets of initial data. In contrast to Wang \cites{Wang_chapter, Wang2007-Asymptotic,Wang2008}, who defines his class of weak solutions by including in the definition the appropriate maximum principle inequality for temperature, and proves their existence by the vanishing viscosity method with fourth order term, we use different procedure. Namely, we consider the class of weak solutions obtainable from the Galerkin procedure with the discretization only of $u$ and $\gamma$ with strongly convergent initial data.
This way we do not have to include both energy inequality and maximum principle in the solution definition. As we cannot prove the solution uniqueness, we study the problem in the framework of Melnik and Valero \cites{Melnik_Valero, Melnik-2008} of the so called \textit{multivalued semiflows}. To deal with the known difficulty of the lack of translation property in the class of Leray--Hopf type weak solutions, we define the new notion of \textit{multivalued eventual semiflow}, by requiring the translation property only after time dependent on the size of initial data. The dynamical system governed by weak solutions for the considered problem satisfies this property, as after time to reach the enstrophy estimate, the solution can no longer branch out, which, in particular, yields the translation property. 

The key assumption for the present result are the relations 
$$
L \geq \frac{16 }{3\pi^2} K
\quad \textrm{and}\quad 
{\rm Pr} \geq 2c_1{\rm Ra}{\left( \max\left\{ 2,\frac{M}{L}\right\}\right) ^{3/2}\sqrt{A}},
$$
where $c_1$ is a universal constant, and $A$ denotes the size of the domain $\Omega$. The first inequality states, in a sense, that the damping in the angular momentum equation, denoted by $L$, must overcome the influence of the coupling between the momentum and angular momentum equations, denoted by $K$. This restriction disappears if $K=0$, that is, in the Newtonian case. On the other hand, the second inequality denotes that the Prandlt number $\Pr$ must be sufficiently large. If $\frac{M}{L}\leq 2$, that is the effect of micropolar inertia is small compared to the micropolar damping, then the  bound reduces to the Newtonian one known from \cites{Wang_chapter, Wang2007-Asymptotic,Wang2008}. If $\frac{M}{L}> 2$, then the inertial effects in the angular momentum equation come into play, and the lower bound on the admissible Prandtl number becomes the increasing function of the ratio $\frac{M}{L}$. Generally the bigger the inertial effects are and the smaller the damping is in the angular momentum equation, the less likely the uniform enstrophy bouds are to hold.  

We stress, that while the present article concerns the three dimensional model, in the two dimensional version, the weak solutions are unique and become instantaneously smooth. Existence and regularity of the global attractor have been studied in \cite{Tarasinska_2006_paper}, while the comparison between micropolar and Newtonian models was undertaken in \cite{KaLaLu_2017_MMN}. There, the estimates from above on Nusselt number in the framework of Constantin and Doering \cites{CD1, CD2} as well as the global attractor fractal dimension have been compared between micropolar and Newtonian fluids, revealing that micropolar fluids tend to be more stable than Newtonian ones. 

The plan of the paper is as follows. In Section~\ref{sect:multivalued-theory} we introduce the notion of the multivalued eventual semiflow and prove some properties of such flows. In Section~\ref{sec:intro} we introduce the scaling, define weak and strong solutions of our problem, as well as prove some of their properties. In Section~\ref{Basic estimates} we provide basic energy and enstrophy estimates, and estimates of the temperature. 
Section~\ref{sect:global} is devoted to proofs of the existence and invariance of the two-space global attractor $\mathcal{A}^K$ corresponding to weak solutions, for every micropolar parameter $K\geq 0$ measuring the deviation of the considered system from the classical Rayleigh--B\'{e}nard problem for the Newtonian fluid.  Moreover, we prove that the semiflow restricted to attractor $\mathcal{A}^K$ is single-valued and governed by strong solutions. Finally, in Section~\ref{sect:upper}, we prove that the global attractors $\mathcal{A}^K$ converge to $\mathcal{A}^0$ upper semicontinuously in Kuratowski (and Hausdorff) sense as $K\to 0$, and that the projection of $\mathcal{A}^0$ on the restricted phase space corresponding to the classical Rayleigh--B\'{e}nard problem is the global attractor for the latter problem, having the invariance property. We stress that the core results are established under the assumption that the Prandtl number is relatively large with respect to the Rayleigh number.

\section{Multivalued eventual semiflows.} \label{sect:multivalued-theory}
In this section we introduce the notion of a {\it multivalued eventual semiflow} which we need in what follows and which would prove useful in many similar situations when we do not know whether a given multivalued semiflow satisfies the translation property, e.g., as in the case of the three--dimensional Navier--Stokes system, but which satisfies the translation property for large times, uniformly for bounded sets of initial data. We prove also relevant properties of such semiflows, in particular that on the existence of the two-space global attractor and of its invariance.

For a complete metric space $(X,\varrho_X)$, we define by $\mathcal{P}(X)$ the family of its nonempty subsets, and by $\mathcal{B}(X)$, the family of nonempty and bounded subsets. We also remind the definition of a Hausdorff semidistance 
$$
\mathrm{dist}_X(A,B) = \sup_{a\in A}\inf_{b\in B}\varrho_X(a,b)\quad \textrm{for}\quad A, B\in \mathcal{P}(X).
$$

The following definition is the slight relaxation of the multivalued semiflow definition due to Melnik and Valero \cites{Melnik_Valero, Melnik-2008}.

\begin{definition}\label{def:eventual}
	Let $(X,\varrho_X)$ be a complete metric space. The family of mappings $\{ S(t) \}_{t\geq 0}$ such that $S(t):X\to \mathcal{P}(X)$ is a multivalued \textit{eventual} semiflow if 
	\begin{itemize}
		\item[(i)] $S(0)v = \{ v \}$ for every $v\in X$.
		\item[(ii)] For every $B\in \mathcal{B}(X)$ bounded there exists time $t_1(B)$ such that for every $t\geq t_1(B)$, $s\geq 0$ there holds $S(s+t)B \subset S(s)S(t)B$.
	\end{itemize}
\end{definition}
The difference between the above definition and that of Melnik and Valero \cites{Melnik_Valero, Melnik-2008} is, that property (ii) is assumed to hold for every $t\geq 0$  in \cites{Melnik_Valero, Melnik-2008} and not for just $t\geq t_1(B)$. As it turns out, it suffices to relax the definition to $t\geq t_1(B)$, and the  result of \cites{Melnik_Valero, Melnik-2008} on the global attractor existence remains valid. 

We pass to the definition of a global attractor for multivalued eventual semiflow. In this definition, firstly, following \cite{Chepyzhov_Conti_Pata}, we do not impose invariance of the global attractor. We only assume that it is a minimal compact attracting set. Secondly, we follow \cite{Babin_Vishik} and assume that our attractor is compact in a ''smaller space'' $Y$ while it attracts bounded sets from a ''bigger'' space $X$. Specifically we make the standing assumptions, that $(X,\varrho_X)$ and $(Y,\varrho_Y)$ are complete metric spaces such that $Y\subset X$ and the identity $i:Y\to X$ is continuous.

\begin{definition}\label{def:xyattr}
Let $\{S(t)\}_{t\geq 0}$ be multivalued eventual semiflow in $X$. The set $\mathcal{A} \subset Y$ is called a $(X,Y)$-global attractor for $\{ S(t) \}_{t\geq 0}$ if 
	\begin{itemize}
		\item[(i)] $\mathcal{A} \in \mathcal{B}(Y)$ is compact in $Y$,
		
		\item[(ii)] there holds $\lim_{t\to \infty}\, \mathrm{dist}_Y(S(t)B,\mathcal{A}) = 0$ for every $B\in \mathcal{B}(X)$,
		
		\item[(iii)] if, for a closed in $Y$ set $\overline{\mathcal{A}}$ there holds $\lim_{t\to \infty}\, \mathrm{dist}_Y(S(t)B,\overline{\mathcal{A}}) = 0$ for every $B\in \mathcal{B}(X)$, then $\mathcal{A} \subset \overline{\mathcal{A}}$.	
	\end{itemize}
\end{definition}

Let $B\in \mathcal{B}(X)$. Assuming that there exists $t_0(B)$ such that for every $t\geq t_0$ there holds $S(t)B\subset Y$ we can define the $\omega$-limit set in $Y$ in the following way.

$$
\omega_Y(B) = \bigcap_{s\geq t_0(B)}\overline{\bigcup_{t\geq s}S(t)B}^Y.
$$

We prove the following theorem.

\begin{theorem}\label{thm:attr}
Let $\{S(t)\}_{t\geq 0}$ be multivalued eventual semiflow in $X$. Assume that the following two conditions hold. 
	\begin{itemize}
		\item[(i)] The family $\{S(t)\}_{t\geq 0}$ is $(X,Y)$- dissipative, i.e., there exists a set $B_0 \in \mathcal{B}(Y)$ such that for every $B\in \mathcal{B}(X)$ there exists $t_0(B)$  such that
		$$
		\bigcup_{t\geq t_0} S(t) B \subset B_0.
		$$
		
		\item[(ii)] The family $\{S(t)\}_{t\geq 0}$ is $Y$- asymptotically compact on $B_0$, i.e., if $t_n\to \infty$ then every sequence $v_n \in S(t_n)B_0$ is relatively compact in $Y$.
	\end{itemize}
	Then the family $\{ S(t) \}_{t\geq 0}$ has a $(X,Y)$-global attractor $\mathcal{A}\subset \overline{B_0}^Y$ which is given by $\mathcal{A} = \omega_Y(B_0)$.  
\end{theorem}
\begin{proof}
	The proof is a simple modification of the standard proof of the global attractor existence \cites{Robinson, Temam}. We provide it for the completeness of the exposition. It is clear that $(X,Y)$-dissipativity implies that $\omega_Y(B_0)\subset \overline{B_0}^Y$ which also implies that $\omega_Y(B_0)$ is bounded in $Y$. We split the rest of the proof in three steps. For a set $B\in \mathcal{B}(X)$ we define $t_2(B) = \max\{ t_0(B), t_1(B) \}$. 
	
	\noindent \textbf{Step 1.} $\omega_Y(B_0)$ is \textit{nonempty} and \textit{compact} in $Y$. As $\omega_Y(B_0)$ is clearly closed in $Y$ we must prove that it is relatively compact in $Y$. Let $v_n \in \omega_Y(B_0)$ be a sequence. Then for every $n$
	$$
	v_n \in \overline{\bigcup_{t\geq t_0(B)+n}S(t)B_0}^Y.
	$$ 
	Hence, there exists the sequence of times $t_n \to \infty$ and the sequence $w_n\in Y$ such that
	$$
	\varrho_Y(v_n, w_n) \leq \frac{1}{n}\quad \textrm{and}\quad w_n \in S(t_n)B_0.
	$$
	The asymptotic compactness implies then that there exists $w\in Y$, such that for a subsequence, still denoted by $n$, there holds $w_n\to w$ in $Y$, whence also $v_n \to w$ in $Y$, which ends the proof of the relative compactness. To prove that it is nonempty, let $v_n \in S(t_n)B_0$ for some sequence $t_n\to \infty$. Asymptotic compactness implies that, for a subsequence, not renumbered, there holds $v_n \to v$ for some $v\in Y$. Then, as for every $s\geq t_0(B_0)$ there exists $n(s)$ such that 
	$$
	\{ v_k \}_{k=n(s)}^\infty \subset \bigcup_{t\geq s} S(t)B_0,
	$$
	it follows that 
	$$
	v \in \overline{\bigcup_{t\geq s} S(t)B_0}^Y\quad \textrm{for every}\quad s
	\geq t_0(B_0),
	$$
	and $v
	\in \omega_Y(B_0)$.
	
	\noindent \textbf{Step 2.} $\omega_Y(B_0)$ \textit{attracts} all sets from $\mathcal{B}(X)$. We should prove that for every $B\in \mathcal{B}(X)$ and $t_n\to \infty$
	$$
	\lim_{n\to\infty}\mathrm{dist}_Y(S(t_n)B,\omega_V(B_0)) = 0.
	$$
	Assume, for the contradiction, that there exists $B\in \mathcal{B}(X)$, $\epsilon >0$ and a subsequence, not renumbered, such that
	$$
	\mathrm{dist}_Y(S(t_n)B,\omega_Y(B_0)) > \epsilon.
	$$
	This means that there exists a sequence $v_n \in S(t_n)B$ such that
	$$
	\mathrm{dist}_Y(v_n,\omega_Y(B_0)) > \epsilon.
	$$ 
	But $v_n \in S(t_n)B$ implies that
	$$
	v_n \in S(t_n-t_2(B)+t_2(B))B \subset S(t_n-t_2(B))S(t_2(B))B \subset S(t_n-t_2(B))B_0. 
	$$
	Asymptotic compactness implies that there exists $v\in Y$ such that, for another subsequence, $v_n \to v$, and 
	$$
	\mathrm{dist}_Y(v,\omega_Y(B_0)) > \epsilon.
	$$
	Hence it must be that $v\in \omega_Y(B_0)$, a contradiction.

	\noindent \textbf{Step 3.} $\omega_Y(B_0)$ is the \textit{smallest} closed set which attracts all sets from $\mathcal{B}(X)$. Let $\overline{\mathcal{A}}$ be such that 
	$$
	\lim_{t\to\infty}\mathrm{dist}_Y(S(t)B,\overline{\mathcal{A}}) = 0\quad \textrm{for every}\quad B\in \mathcal{B}(X).
	$$
	Take $v\in \omega_Y(B_0)$. There exist sequences $t_n\to \infty$ and $v_n \in S(t_n)B_0$ such that $v_n\to v$ in $Y$. This means that
	$$
	\lim_{n\to\infty}\mathrm{dist}_Y(v_n,\overline{\mathcal{A}}) = 0,
	$$
	whereas $v\in \overline{\mathcal{A}}$, as $\overline{\mathcal{A}}$ is a closed set. The proof is complete. 
\end{proof}

We pass to the global attractor invariance. Namely, the following theorem holds.
\begin{theorem}\label{thm:invariance_abstract}
	Assume, in addition to the assumptions of Theorem \ref{thm:attr}, that for every $t\geq 0$ the restriction $S(t)|_Y$ is a single valued semigroup of $(Y,Y)$-continuous maps. Then, the  $(X,Y)$-global attractor $\mathcal{A}$ is an invariant set, that is $S(t)\mathcal{A} = \mathcal{A}$ for every $t\geq 0$.     
\end{theorem}
\begin{proof}
	Since assumptions of Theorem \ref{thm:attr} imply that the single valued semigroup $\{ S(t)|_Y \}_{t\geq 0}$ has a bounded absorbing set in $Y$ and is asymptotically 
compact on $Y$, the result follows from a very well known abstract theorem on the global attractor existence, cf., e.g., \cites{Robinson, Temam}. 
\end{proof}
We conclude this section with results on the relation between global attractors and complete (eternal) bounded trajectories.
\begin{definition}
	The function $u:\mathbb{R} \to X$ is a $X$-bounded (respectively, $Y$-bounded) eternal trajectory if for every $t\in \mathbb{R}$ and every $s>t$ there holds $u(s) = S(s-t)u(t)$ and the set $\{ u(t) \}_{t\in \mathbb{R}}$ is bounded in $X$ (respectively, $Y$).
\end{definition}

Since the family $\{ S(t)|_Y \}_{t\geq 0}$ is a semigroup on $Y$ with the global attractor $\mathcal{A}$, the following theorem follows from the well known result for the semigroups, see, e.g., \cite[Lemma 1.4 and Theorem 1.7]{CLR2013}. 
\begin{theorem}\label{thm:complete1}
	Under assumptions of Theorem \ref{thm:invariance_abstract}, for every $v\in \mathcal{A}$ and for every $t\in \mathbb{R}$ there exists the eternal $Y$-bounded (and hence also $X$-bounded) trajectory $u:\mathbb{R} \to Y$ such that $u(t)=v$ and $u(s)\in \mathcal{A}$ for every $s\in \mathbb{R}$. 
\end{theorem}	
On the other hand if the function $u:\mathbb{R} \to X$ is a complete $X$-bounded trajectory, then the next theorem shows that  $u$ has values in $Y$ and $u(t)\in \mathcal{A}$ for every $t\in \mathbb{R}$. The proof mostly follows the lines of \cite[Theorem 1.7]{CLR2013}, with the modification that we consider only $X$-bounded and not $Y$-bounded trajectories. 

\begin{theorem}\label{thm:complete_2}
	Under assumptions of Theorem \ref{thm:invariance_abstract}, for every complete $X$-bounded trajectory $u:\mathbb{R} \to X$ there holds $u(t) \in \mathcal{A}\,$ for every $t\in \mathbb{R}$. 
\end{theorem}
\begin{proof}
	Denote $B = \bigcup_{t\in \mathbb{R}}\{ u(t) \}$. This is a bounded set in $X$. We show that this set is actually bounded in $Y$.
	Indeed, take $s$ such that $t-s\geq t_0(B)$. It follows that
	$$
	u(t) = S(t-s)u(s) \in \bigcup_{r\geq t_0} S(r)B \subset B_0.
	$$
	This means that $u$ is $Y$-bounded. 
	Now, for $s\in \mathbb{R}$ and $t\geq s+t_0(B)$ there holds
	$$
	\mathrm{dist}_Y(u(t),\mathcal{A}) \leq \mathrm{dist}_Y(S(t-s)u(s),\mathcal{A}) \leq \mathrm{dist}_Y(S(t-s)B,\mathcal{A}).
	$$
	Passing with $s$ to $-\infty$ we deduce that
	$$
	\mathrm{dist}_Y(u(t),\mathcal{A}) = 0,
	$$
	whence $u(t) \in Y$ as $\mathcal{A}$ is $Y$-closed.
\end{proof}

\section{Problem formulation, scalings and preliminaries.}\label{sec:intro}
  \noindent \textbf{Problem formulation.} Let $\Omega = (0,L_{x_1})\times (0,L_{x_2})\times(0,h)\subset \R^3$. The boundary of $\Omega$ is divided into three parts $\partial\Omega = \overline{\Gamma_B}\cup\overline{\Gamma_T}\cup \overline{\Gamma_L}$, where $\Gamma_B = (0,L_{x_1})\times (0,L_{x_2})\times \{0\}$ is the bottom, $\Gamma_T = (0,L_{x_1})\times (0,L_{x_2})\times \{h\}$ is the top, and $\Gamma_L$ is the lateral boundary.
  
We consider the following Rayleigh--B\'{e}nard problem for micropolar fluid, where $(x,t)\in \Omega\times (0,\infty)$.
\begin{align}
& u_t + (u\cdot \nabla) u - (\nu+\nu_r) \Delta u + \varrho_0^{-1}\nabla p = 2\nu_r {\rm rot}\, \gamma + e_3 \overline{\alpha}g T,\label{eq:1a}\\
& {\rm div}\, u = 0,\label{eq:2a}\\
& j(\gamma_t + (u\cdot \nabla)\gamma) - \alpha\Delta \gamma - \beta\nabla{\rm div}\,\gamma
 + 4\nu_r \gamma = 2\nu_r {\rm rot}\, u,\label{eq:3a}\\
& T_t + u\cdot \nabla T - \chi \Delta T = 0.\label{eq:4a} 
\end{align}
The unknowns in the above equations are the velocity field $u:\Omega\times (0,\infty)\to \R^3$, the pressure $p:\Omega\times (0,\infty)\to \R$, the temperature $T:\Omega\times (0,\infty)\to \R$, and the microrotation field $\gamma:\Omega\times (0,\infty)\to \R^3$. 
We assume the boundary conditions $u = 0$ and $\gamma = 0$ on $\Gamma_B\cup \Gamma_T$, $T=0$ on $\Gamma_T$, and $T=T_B$ on $\Gamma_B$. On $\Gamma_L$ we impose the periodic conditions on the functions $u, T, \gamma$ and their normal derivatives as well as the pressure $p$ such that all boundary integrals on $\Gamma_L$ in the weak formulation cancel. All physical constants $\nu, \nu_r, \varrho_0, \overline{\alpha}, g, j, \alpha, \chi, T_B$ are  assumed to be positive numbers, save for $\nu_r$, where we allow for the case $\nu_r = 0$, in which the micropolar model reduces to the Boussinesq one, $e_3=(0,0,1)$. Finally we impose the initial conditions 
$$
u(x,0)=u_0(x), \quad T(x,0)=T_0(x), \quad \gamma(x,0)=\gamma_0(x)\quad\text{for}\quad x\in \Omega.
$$
To reduce the number of parameters in the models and make the variables independent on the measurement units we introduce the scaling of unknowns, corresponding to the one for Newtonian fluids of \cite{DG} and \cite{Wang2008}. In the scaling we use the following dimensionless numbers.
\begin{itemize}
	\item Rayleigh number ${\rm Ra}=\frac{\overline{\alpha}gT_Bh^3}{\nu\chi}$.
	\item Prandtl number ${\rm Pr}=\frac{\nu}{\chi}$. 
	\item Grashof number ${\rm Gr} = \frac{\rm Ra}{\rm Pr} = \frac{\overline{\alpha}gT_Bh^3}{\nu^2}$.
	\item Microrotation ratio $N=\frac{\nu_r}{\nu+\nu_r}$. Clearly, as $\nu>0, \nu_r\geq0$, we have $0\leq N<1$. For the case $N=0$ the equations decouple, and we get the classical Raleigh--B\'{e}nard problem and the additional equation for the microrotation. For the ease of notation, in place of $N$ we use the constant $K = \frac{N}{1-N} = \frac{\nu_r}{\nu}$, see \cite{PayneStraughan}. Then $K\geq 0$. 
	\item Micropolar damping $L = \frac{\alpha}{h^2\nu}$.
	\item Micropolar inertia $M = \frac{j}{h^2}$.
	\item Second micropolar viscosity $G=\frac{\beta}{h^2\nu}$.
\end{itemize}
The scaled variables and unknowns are given as
\begin{align}\label{eq:scaling_1}
	&x = h x',\quad t=\frac{h^2}{\chi }t', \quad u = \frac{\chi}{h}u',\quad p = \frac{\varrho_0 \chi^2}{h^2}p',\quad T= T_B T',\quad \gamma = \frac{\chi}{h^2}\gamma'.
\end{align}
	As $x=h x'$, the set $\Omega$ is scaled to $\Omega'=(0,L_{x_1}/h)\times(0,L_{x_2}/h)\times (0,1)$. Since from now on we will work only with the scaled equations and variables, for the notation convenience we will always drop the primes and use the same symbols for scaled variables as for the nonscaled ones. Note that the periodic conditions on $\Gamma_L$ for all variables are always transformed to the corresponding periodic conditions on the scaled lateral boundary. Similarly, the homogeneous Dirichlet conditions are also always transformed to the corresponding homogeneous condition. The only boundary condition which scales nontrivially is the condition $T=T_B$ on $\Gamma_B$. Note that the scaled $\Gamma_B$ is given by 
$(0,L_{x_1}/h)\times(0,L_{x_2}/h)\times \{0\}$. We shall denote $\epsilon=1/{{\rm Pr}}$. 
System \eqref{eq:1a}--\eqref{eq:4a} in dimensionless variables takes the following form. For $(x,t)\in \Omega\times (0,\infty)$,
\begin{align}
	& \epsilon(u_t + (u\cdot \nabla) u) - (1+K)\Delta u + \epsilon\nabla p = 2K {\rm rot}\, \gamma + e_3  {\rm Ra} T,\label{eq:1d}\\
	& {\rm div}\, u = 0,\label{eq:2d}\\
	& \epsilon M(\gamma_t + (u\cdot \nabla)\gamma) - L\Delta \gamma - 
	G\nabla{\rm div}\,\gamma
	+ 4K \gamma = 2K {\rm rot}\, u,\label{eq:3d}\\
	& T_t + u\cdot \nabla T - \Delta T = 0.\label{eq:4d} 
\end{align}
The boundary condition for the temperature on $\Gamma_B$ is now given as 
\begin{equation}\label{eq:bc_T}
T(x_1,x_2,0,t) = 1\quad {\rm for}\quad (x_1,x_2)\in (0,L_{x_1}/h)\times(0,L_{x_2}/h).
\end{equation} 
 \noindent \textbf{Background temperature.} To make the boundary condition \eqref{eq:bc_T} homogeneous we will introduce the background temperature, which is a function $\tau:[0,1]\to \mathbb{R}$ of the variable $x_3$, such 
that $\tau(1) = 0$ and $\tau(0) = 1$. We will replace the temperature with a new unknown $\theta$ such that
$$
T(x_1,x_2,,x_3,t) = \theta(x_1,x_2,x_3,t) + \tau(x_3),
$$
and solve the problem for $\theta$. Such translation  
will make the boundary conditions for $\theta$ homogeneous both on the top and bottom boundary, i.e. 
$$
\theta(x_1,x_2,1,t) = \theta(x_1,x_2,0,t) = 0\quad {\rm for}\quad (x_1,x_2)\in (0,L_{x_1}/h)\times(0,L_{x_2}/h).
$$ 
 After introducing the background temperature, system \eqref{eq:1d}--\eqref{eq:4d} takes the form 
\begin{align}
& \epsilon(u_t + (u\cdot \nabla) u) - (1+K)\Delta u + \nabla \bar{p} = 2K {\rm rot}\, \gamma + e_3  {\rm Ra} \theta,\label{eq:1e}\\
& {\rm div}\, u = 0,\label{eq:2e}\\
& \epsilon M(\gamma_t + (u\cdot \nabla)\gamma) - L\Delta \gamma - G\nabla{\rm div}\,\gamma
+ 4K \gamma = 2K {\rm rot}\, u,\label{eq:3e}\\
& \theta_t + u\cdot \nabla \theta + u_3 \tau' - \Delta \theta - \tau'' = 0,\label{eq:4e} 
\end{align}
where the pressure $p$ has been replaced by $\bar{p} = \epsilon p + {\rm Ra} \int_0^{x_3} \tau(r)\, dr$. In the sequel we will just write $p$ in place of $\bar{p}$. 
Various choices of $\tau$ are possible, the simplest choice is $\tau (x_2) = 1-x_2$, and with such choice system \eqref{eq:1e}--\eqref{eq:4e} takes the form
\begin{align}
& \epsilon(u_t + (u\cdot \nabla) u) - (1+K)\Delta u + \nabla p = 2K {\rm rot}\, \gamma + e_3  {\rm Ra} \theta,\label{eq:1f}\\
& {\rm div}\, u = 0,\label{eq:2f}\\
& \epsilon M(\gamma_t + (u\cdot \nabla)\gamma) - L\Delta \gamma - G\nabla{\rm div}\,\gamma
+ 4K\gamma = 2K {\rm rot}\, u,\label{eq:3f}\\
& \theta_t + u\cdot \nabla \theta - \Delta \theta  = u_3.\label{eq:4f} 
\end{align}
Together with the above system we shall consider the following system
\begin{align}
& \epsilon(u_t + (u\cdot \nabla) u) - \Delta u + \nabla p = e_3  {\rm Ra} \theta,\label{eq:1f_kz}\\
& {\rm div}\, u = 0,\label{eq:2f_kz}\\
& \epsilon M(\gamma_t + (u\cdot \nabla)\gamma) - L\Delta \gamma - G\nabla{\rm div}\,\gamma
 = 0,\label{eq:3f_kz}\\
& \theta_t + u\cdot \nabla \theta - \Delta \theta  = u_3.\label{eq:4f_kz}
\end{align}
formally corresponding to the Newtonian fluid, that is $K=0$. Then, \eqref{eq:1f_kz}, \eqref{eq:2f_kz}, and \eqref{eq:4f_kz} constitute the well known Boussinesq system, while the equation \eqref{eq:3f_kz} can be independently solved for $\gamma$ once the solution of the system  of the remaining three equations is known. We stress that all definitions and results stated below are valid also for the case $K=0$.


\vspace{0.3cm}

\noindent \textbf{Definition of the weak solution.}
We introduce some notation. By $(\cdot,\cdot)$ we will denote the scalar product in $L^2(\Omega)$ or, depending on the context, in $L^2(\Omega)^3$. Let $\widetilde{V}$ be the space of divergence-free functions which are restrictions to $\overline{\Omega}$ of functions from $C^\infty(\mathbb{R}^2\times [0,1])^3$  which are equal to zero on $\mathbb{R}^2\times \{0,1\}$, and, together with all their derivatives, $(L_{x_1}/h,L_{x_1}/h)$-periodic with respect to the first two variables. Similarly, by $\widetilde{W}_k$, $k=1,3$, we define the space of all functions which are restrictions to $\overline{\Omega}$ of functions from $C^\infty(\mathbb{R}^2\times [0,1])^k$ which are $(L_{x_1}/h,L_{x_1}/h)$-periodic with all derivatives with respect to the first two variables and satisfying the homogeneous Dirichlet boundary conditions on bottom and top boundaries $\mathbb{R}^2\times \{0,1\}$. Define the spaces
$$
V = \{ \text{\rm  closure of } \widetilde{V}\  \text{\rm  in } H^1(\Omega)^3\, \}\quad \text{\rm and}\quad H = \{ \text{\rm  closure of } \widetilde{V}\  \text{\rm  in } L^2(\Omega)^3\, \},
$$    
$$
W_k = \{ \text{\rm  closure of } \widetilde{W}_k\  \text{\rm  in } H^1(\Omega)^k\, \}\quad \text{\rm and}\quad E_k = \{ \text{\rm  closure of } \widetilde{W}_k\  \text{\rm  in }
 L^2(\Omega)^k\, \}.
$$    
Moreover, let $P:L^2(\Omega)^3\to H$ be the Leray--Helmholz projection. The duality pairings between $V$ and its dual $V'$ as well as between $W_k$ and its dual $W_k'$ will be denoted by $\langle\cdot, \cdot\rangle$, and the scalar product on $E_k$ and $H$ will be denoted by $(\cdot,\cdot)$.

We will consider the eigenfunctions and eigenvalues of the three operators: the scalar operator $-\Delta$, the vector valued three--dimensional operator $-\Delta$ and the three-dimensional Stokes operator with the considered boundary conditions. All these operators have the sequences of positive eigenvalues going to infinity with the corresponding eigenfunctions being smooth and constituting the complete orthonormal sequences in $E_1$, $E_3$, and $H$, respectively. By $D_1(-\Delta)$, $D_3(-\Delta)$, and $D(-P \Delta)$ we will denote, respectively, domains of one and three-dimensional Laplacian and the Stokes operator equipped with $H^2$ norms. It is not hard to verify that the first eigenvalue of all three operators coincides and is equal to  $\lambda_1 = \pi^2$ \cite{Rummler_2}. Hence, the following Poincar\'{e} inequalities hold
\begin{align*}
& \pi^2\|u\|^2 \leq \|\nabla u\|^2_2 \quad \textrm{for every}\quad u\in V,\\
& \pi^2\|\nabla u\|^2 \leq \|P\Delta u\|^2_2 \quad \textrm{for every}\quad u\in D(-P\Delta),\\
& \pi^2\|\gamma\|^2 \leq \|\nabla \gamma\|^2_2 \quad \textrm{for every}\quad u\in W_3,\\
& \pi^2\|\nabla \gamma\|^2 \leq \|\Delta \gamma\|^2_2 \quad \textrm{for every}\quad u\in D_3(-\Delta),\\
& \pi^2\|\theta\|^2 \leq \|\nabla \theta\|^2_2 \quad \textrm{for every}\quad u\in W_1,\\
& \pi^2\|\nabla \theta\|^2 \leq \|\Delta \theta\|^2_2 \quad \textrm{for every}\quad u\in D_1(-\Delta).
\end{align*}
By $E_1^n$ and $E_3^n$ we will denote the spaces spanned by the first $n$ eigenfunctions of the scalar and vector $-\Delta$, respectively, and by $H^n$ the space spanned by the first $n$ eigenfunctions of the considered Stokes operator. 

The weak solution to the above problem will be considered as the limit of the Galerkin problems. But, as we need the solutions of the Galerkin problem to satisfy the maximum principle for the temperature we only discretize the velocity $u$ and the microrotation $\gamma$ and keep the original PDE for the temperature. In the following definition of the Galerkin problem by $AC(I;X)$ we denote the space of absolutely continuous functions defined on the time interval $I$ with values in a Banach space $X$.

\begin{definition} \label{def-weak-sol-evol-galerkin} 
	The ($u,\gamma$)-Galerkin problem for equations \eqref{eq:1f}--\eqref{eq:4f} is defined as follows. Let $(u^n_0, \gamma^n_0,\theta^n_0)\in H^n\times E^n_3\times E^n_1$ be some approximation of $(u_0,\gamma_0,\theta_0) \in H \times E_3 \times E_1$.
	Find
	\begin{itemize}
		\item $u^n\in AC([0,\infty);H^n)$, $\,\,\, u^n(0)=u^n_0\in H^n$,
		\item $\gamma^n\in AC([0,\infty);E_3^n)$, $\,\,\,\gamma^n(0) = \gamma^n_0 \in E_3^n\,\,\,$,
		\item $\theta^n\in L^2_{loc}([0,\infty);W_1)$,  
		with $\theta^n_t\in L^{2}_{loc}([0,\infty);W_1')$ and $\theta^n(0) = \theta^n_0\in E_1$,
	\end{itemize}
	such that for all test functions $v\in H^n$, $\xi \in E_3^n$, $\eta\in W_1$ and a.e. $t>0$ there holds 
	\begin{align}
	& \epsilon\left((u^n_t(t),v) + ((u^n(t)\cdot \nabla) u^n(t),v)\right) + (1+K)(\nabla u^n(t),\nabla v) = 2K ({\rm rot}\, \gamma^n(t),v) + {\rm Ra} (\theta^n(t),v_3),\label{eq:1f_w_galer}\\
	& \epsilon M\left((\gamma^n_t(t),\xi) + ((u^n(t)\cdot \nabla)\gamma^n(t),\xi)\right) + L(\nabla \gamma^n(t),\nabla \xi) + G({\rm div}\,\gamma^n(t), {\rm div}\,\xi) +
	4 K (\gamma^n(t),\xi)\nonumber\\
	&\qquad = 2K ({\rm rot}\, u^n(t),\xi),\label{eq:3f_w_galer}\\
	& \langle \theta^n_t(t), \eta \rangle + (u^n(t)\cdot \nabla \theta^n(t),\eta) + (\nabla  \theta^n(t),\nabla \eta)  = (u^n_3(t),\eta).\label{eq:4f_w_galer} 
	\end{align}
\end{definition}
The proof of the following result is standard and so we skip it. The argument follows for example by the Galerkin method applied to the equation for $\theta^n$, cf. \cite{Foias_MRT-NSE_Turb}.
\begin{lemma}
For every $u^n_0 \in H^n, \gamma^n_0\in E^n_3$, and $\theta_0 \in E_1$ the $(u,\gamma)$-Galerkin problem given in Definition  \ref{def-weak-sol-evol-galerkin} has a unique solution. 
\end{lemma}
We define the weak solution to the above problem as the approximative limit of subsequences of the $(u,\gamma)$-Galerkin problems.
\begin{definition} \label{def-weak-sol-evol} 
	Let
	$$
	u_0 \in H, \quad \gamma_0\in E_3,\quad\textrm{and}\quad \theta_0 \in E_1. 
	$$
	The triple of functions $(u,\gamma,\theta)$ such that 
	\begin{itemize}
		\item $u\in L^2_{loc}([0,\infty);V)\cap C_w([0,\infty);H)$, with $u_t\in L^{4/3}_{loc}([0,\infty);V')$ and $u(0)=u_0$,
		
		\item $\gamma\in L^2_{loc}([0,\infty);W_3)\cap C_w([0,\infty);E_3)$, with $\gamma_t\in L^{4/3}_{loc}([0,\infty);W_3')$ and $\gamma(0) = \gamma_0$,
		
		\item $\theta\in L^2_{loc}([0,\infty);W_1)\cap C_w([0,\infty);E_1)$,  
		with $\theta_t\in L^{4/3}_{loc}([0,\infty);W_1')$, and $\theta(0) = \theta_0$,
		
	\end{itemize}
is called a weak solution of the problem \eqref{eq:1f}--\eqref{eq:4f}  if for all test functions $v\in V$, $\xi \in W_3$, $\eta\in W_1$ and a.e. $t>0$ there holds 
	\begin{align}
	& \epsilon\left(\langle u_t(t),v\rangle + ((u(t)\cdot \nabla) u(t),v)\right) + (1+K)(\nabla u(t),\nabla v) = 2K ({\rm rot}\, \gamma(t),v) + {\rm Ra} (\theta(t),v_3),\label{eq:1f_w} \\
	& \epsilon M\left(\langle \gamma_t(t),\xi \rangle + (u(t)\cdot \nabla\gamma(t),\xi)\right) + L(\nabla \gamma(t),\nabla \xi) + G({\rm div}\,\gamma, {\rm div}\,\xi) +
	4 K (\gamma(t),\xi)\nonumber\\
	&\qquad = 2K ({\rm rot}\, u(t),\xi),\label{eq:3f_w}\\
	& \langle \theta_t(t), \eta\rangle  + (u(t)\cdot \nabla \theta(t),\eta) + (\nabla  \theta(t),\nabla \eta)  = (u_3(t),\eta)\label{eq:4f_w} 
	\end{align}
	and
	if $(u,\gamma,\theta)$ is the limit of approximative problems in the sense that there exist the sequences of initial 
	data
	\begin{align*}
	& H^n\ni u_0^n\to u_0\in H\quad \textrm{strongly in}\ H,\\
	& E_3^n\ni \gamma_0^n\to \gamma_0\in E_3 \quad \textrm{strongly in}\ E_3,\\
	& E_1\ni \theta_0^n\to \theta_0\in E_1 \quad \textrm{strongly in}\ E	
	\end{align*}
	such that if $(u^n,\gamma^n,\theta^n)$ are corresponding solutions to ($u,\gamma$)-Galerkin problems with initial data $(u_0^n,\gamma_0^n,\theta_0^n)$, then, for a subsequence of indexes denoted by $n_k$ there holds
	\begin{align}
	& u^{n_k}\to u \quad \textrm{weakly in}\ L^2_{loc}([0,\infty);V)\ \textrm{and weakly}-*\ \textrm{in}\ \ L^\infty_{loc}([0,\infty);H),\label{conv_1}\\
		& \gamma^{n_k}\to \gamma \quad \textrm{weakly in}\ L^2_{loc}([0,\infty);W_3)\ \textrm{and weakly}-*\ \textrm{in}\ \ L^\infty_{loc}([0,\infty);E_3),\\
		& \theta^{n_k}\to \theta \quad \textrm{weakly in}\ L^2_{loc}([0,\infty);W_1)\ \textrm{and weakly}-*\ \textrm{in}\ \ L^\infty_{loc}([0,\infty);E_1).\label{conv_3}
	\end{align} 
\end{definition}
The following result on the existence of the weak solution given in Definition \ref{def-weak-sol-evol} is standard, cf. \cite{Foias_MRT-NSE_Turb}, so we omit its proof.
\begin{lemma}\label{lem:existence}
For every initial data $u_0 \in H, \gamma_0\in E_3$, and $\theta_0 \in E_1$ there exists a weak solution to problem \eqref{eq:1f}--\eqref{eq:4f} given by Definition \ref{def-weak-sol-evol}.  
\end{lemma}

 \noindent \textbf{Some basic properties of the weak solution.}
 We prove some basic properties of the weak solutions. 
 \begin{lemma}\label{lemma:convergence}
 	Suppose that $(u^{n_k},\gamma^{n_k},\theta^{n_k})$ is the subsequence of 
 	$u,\gamma$-Galerkin problems such that the convergences \eqref{conv_1}--\eqref{conv_3} hold. Then there also hold the following convergences
 	\begin{align*}
 		& u^{n_k}(t)\to u(t) \quad \textrm{weakly in}\ H\ \textrm{for every}\ t \geq 0 \textrm{ and strongly in}\ H\ \textrm{for a.e.}\ t > 0\ \textrm{and for}\ t=0,\\
 			& u^{n_k}\to u \quad \textrm{strongly in}\ L^{2}_{loc}([0,\infty);H),\\
 		& u^{n_k}_t\to u_t \quad \textrm{weakly in}\ L^{4/3}_{loc}([0,\infty);V'),\\
 		& \gamma^{n_k}(t)\to \gamma(t) \quad \textrm{weakly in}\ E_3\ \textrm{for every}\ t \geq 0 \textrm{ and strongly in}\ E_3\ \textrm{for a.e.}\ t > 0\ \textrm{and for}\ t=0,\\
 		 			& \gamma^{n_k}\to \gamma \quad \textrm{strongly in}\ L^{2}_{loc}([0,\infty);E_3),\\
 		 		& \gamma^{n_k}_t\to \gamma_t \quad \textrm{weakly in}\ L^{4/3}_{loc}([0,\infty);W_3'),\\
 		 		& \theta^{n_k}(t)\to \theta(t) \quad \textrm{weakly in}\ E_1\ \textrm{for every}\ t \geq 0 \textrm{ and strongly in}\ E_1\ \textrm{for a.e.}\ t > 0\ \textrm{and for}\ t=0,\\
 		 					& \theta^{n_k}\to \theta \quad \textrm{strongly in}\ L^{2}_{loc}([0,\infty);E_1),\\
 			& \theta^n_t\to u_t \quad \textrm{weakly in}\ L^{4/3}_{loc}([0,\infty);W_1').
 	\end{align*}
 \end{lemma}
\begin{proof}
	The result is standard so, again, we omit the details of the proof. The convergence of time derivatives follows from the fact that $(u^{n_k},\gamma^{n_k},\theta^{n_k})$ satisfy the equations \eqref{eq:1f_w_galer}--\eqref{eq:4f_w_galer} and from the bounds implied by the convergences \eqref{conv_1}--\eqref{conv_3}. Strong convergence in $L^2_{loc}([0,\infty);H)$, $L^2_{loc}([0,\infty);E_3)$, and $L^2_{loc}([0,\infty);E_1)$ follows from the Aubin--Lions lemma. This convergence implies the strong convergence in $H$, $E_1$, $E_3$ for a.e. $t$. The strong convergence at $t=0$ follows from the way the initial data are defined for the approximative problems. Obtained convergences also imply the pointwise in time weak convergences in $H$, $E_1$, and $E_3$ for every $t\geq 0$. 
\end{proof}

Similarly as for the Leray--Hopf weak solutions of the 3D Navier--Stokes equations we will need that our weak solutions satisfy the energy inequalities. However, we do not impose these inequalities in the definition of the weak solution. Rather that that, they will follow from the corresponding energy equations for the ($u,\gamma$)-Galerkin problems.

 \begin{lemma}Suppose that the triple $(u,\gamma,\theta)$ is the weak solution given by Definition \ref{def-weak-sol-evol}. Then for almost every $t_0 \geq 0$ (including $t_0=0$) and for every $t > t_0$ there hold the inequalities
 	\begin{align}\label{ine-u:4f_w} 
 	&\epsilon\|u(t)\|_2^2 + 2(1+K)\int_{t_0}^t\|\nabla u(s)\|_2^2ds \leq \epsilon\|u(t_0)\|_2^2 +
 	4K\int_{t_0}^t ({\rm rot}\, \gamma(s),u(s))ds + 2{\rm Ra} \int_{t_0}^t(\theta(s),u_3(s))ds,\\
 	&\epsilon M\|\gamma(t)\|_2^2 + 2L\int_{t_0}^t\|\nabla \gamma(s)\|_2^2ds + 2G\int_{t_0}^t\|\mathrm{div}\, \gamma(s)\|_2^2ds + 8K\int_{t_0}^t \|\gamma(s)\|_2^2\, ds\nonumber \\
 	&\qquad \qquad  \leq \epsilon M\|\gamma(t_0)\|_2^2 +
4K\int_{t_0}^t ({\rm rot}\, u(s),\gamma(s))ds,\label{ine-omega:4f_w} \\
 	&\|\theta(t)\|_2^2 + 2\int_{t_0}^t\|\nabla \theta(s)\|_2^2ds \leq  \|\theta(t_0)\|_2^2 +
\int_{t_0}^t (u_3(s),\theta(s))ds.\label{ine-theta:4f_w}
 	\end{align} 
 \end{lemma}
\begin{proof}
	The proof is standard and follows the same argument as the proof of energy inequalities in the Leray--Hopf weak solution of the Navier Stokes equation. For example, to get \eqref{ine-u:4f_w} we need to test \eqref{eq:1f_w_galer} with $u^n(t)$. After integration over the interval $(t_0,t)$ it follows that
	\begin{align*}
		& \epsilon\|u^n(t)\|_2^2 + 2(1+K)\int_{t_0}^t\|\nabla u^n(s)\|_2^2ds \\
		&\qquad = \epsilon\|u^n(t_0)\|_2^2 +
	4K\int_{t_0}^t ({\rm rot}\, \gamma^n(s),u^n(s))ds + 2{\rm Ra} \int_{t_0}^t(\theta^n(s),u^n_3(s))ds.
	\end{align*}
	Now, the convergences \eqref{conv_1}--\eqref{conv_3} and Lemma \ref{lemma:convergence} as well as the sequential weak lower semicontinuity of the norms imply the desired result. 
\end{proof}

\noindent \textbf{Strong solution and weak-strong uniqueness.} We pass to the definition of the strong solution of the finite Prandtl number problem. 

\begin{definition} \label{def-strong}
	Let
$$
u_0 \in V, \quad \gamma_0\in W_3,\quad\textrm{and}\quad \theta_0 \in W_1. 
$$
The triple of functions $(u,\gamma,\theta)$ such that 
\begin{itemize}
	\item $u\in L^2_{loc}([0,\infty);D(-P \Delta))\cap C([0,\infty);V)$, with $u_t\in L^2_{loc}([0,\infty);H)$ and $u(0)=u_0$,
	
	\item $\gamma\in L^2_{loc}([0,\infty);D_3(-\Delta))\cap C([0,\infty);W_3)$, with $\gamma_t\in L^2_{loc}([0,\infty);E_3)$ and $\gamma(0) = \gamma_0$,
	
	\item $\theta\in L^2_{loc}([0,\infty);D_1(-\Delta))\cap C([0,\infty);W_1)$,  
	with $\theta_t\in L^2_{loc}([0,\infty);E_1)$, and $\theta(0) = \theta_0$,
	
\end{itemize}
is called a strong solution of the problem \eqref{eq:1f}--\eqref{eq:4f}  if for all test functions $v\in H$, $\xi \in E_3$, $\eta\in E_1$ and a.e. $t>0$ there holds 
\begin{align}
& \epsilon\left(( u_t(t),v ) + ((u(t)\cdot \nabla) u(t),v)\right) + (1+K)(-P \Delta u(t), v) = 2K ({\rm rot}\, \gamma(t),v) + {\rm Ra} (\theta(t),v_3), \label{eq:strong_1}\\
& \epsilon M\left(( \gamma_t(t),\xi ) + (u(t)\cdot \nabla\gamma(t),\xi)\right) + L(-\Delta \gamma(t),\xi) + G(-\nabla {\rm div}\,\gamma(t), \xi) +
4 K (\gamma(t),\xi)\nonumber\\
&\qquad = 2K ({\rm rot}\, u(t),\xi),\label{eq:strong_2}\\
& ( \theta_t(t), \eta )  + (u(t)\cdot \nabla \theta(t),\eta) + (-\Delta  \theta(t), \eta)  = (u_3(t),\eta).\label{eq:strong_3}
\end{align}
\end{definition}

\begin{remark}
	We stress  that we do not know on the existence of the strong solution \textit{for every} initial data $
	u_0 \in V,  \gamma_0\in W_3,\textrm{and}\ \theta_0 \in W_1. 
	$ If, however, the constants of the problem satisfy some restriction which will be given later, and the initial data is sufficiently small, such strong solution  always exists. We will prove that  it always exists on the global attractor for the weak solutions. 
\end{remark}
The next result establishes the weak-strong uniqueness property of strong and weak solutions.

\begin{lemma}\label{lemma:weak-strong}
	If $(u,\gamma,\theta)$ is a strong solution then it is also a weak solution and it is moreover unique in the class of the weak solutions. 
\end{lemma}
\begin{proof}
	The proof follows the lines of the weak-strong uniqueness proof of the Leray--Hopf weak solutions. First we observe that every weak solution given by Definition \ref{def-weak-sol-evol-galerkin} satisfies the regularity
	$$
	u \in L^2_{loc}([0,\infty);D(-P\Delta)'), \quad \gamma \in L^2_{loc}([0,\infty);D_3(-\Delta)'),\quad \textrm{and}\quad \theta \in L^2_{loc}([0,\infty);D_1(-\Delta)'). 
	$$
	Assume that $u_0 \in V, \gamma_0 \in V_3, \theta_0\in V_1$ and the triples $(u_1,\gamma_1, \theta_1)$ and $(u_2,\gamma_2, \theta_2)$ are, respectively, weak and strong solutions with these initial data. Denote $(v,\psi, \eta) = (u_1,\gamma_1, \theta_1) - (u_2,\gamma_2, \theta_2)$.
	Then
	$$
	\|v(t)\|^2_{L^2} = \|u_1(t)\|^2_{L^2} - \|u_2(t)\|^2_{L^2} - 2 (u_2(t),v(t)). 
	$$
	Now 
	$$
	(u_2(t),u_1(t)) = \|u_0\|_{L^2}^2 + \int_{0}^t \, ((u_2)_t(s), u_1(s)) + \langle (u_1)_t(s), u_2(s) \rangle_{D(-P\Delta)'\times D(-P\Delta)} ds,
	$$
	and 
	$$
	(u_2(t),u_2(t)) = \|u_0\|_{L^2}^2 + \int_{0}^t \, ((u_2)_t(s), u_2(s)) + \langle (u_2)_t(s), u_2(s) \rangle_{D(-P\Delta)'\times D(-P\Delta)} ds.
	$$
	This means that
	$$
	\|v(t)\|^2_{L^2} = \|u_1(t)\|^2_{L^2} - \|u_2(t)\|^2_{L^2} -2 \int_{0}^t \, ((u_2)_t(s), v(s)) + \langle v_t(s), u_2(s) \rangle_{D(-P\Delta)'\times D(-P\Delta)} ds.
	$$
	Now, \eqref{ine-u:4f_w} implies that
	$$
\|u_1(t)\|_2^2 + \frac{2(1+K)}{\epsilon}\int_{0}^t\|\nabla u_1(s)\|_2^2ds \leq \|u_0\|_2^2 +
	\frac{4K}{\epsilon}\int_{0}^t ({\rm rot}\, \gamma_1(s),u_1(s))ds + \frac{2{\rm Ra}}{\epsilon} \int_{0}^t(\theta_1(s),(u_1)_3(s))ds.
	$$
	On the other hand \eqref{eq:strong_1} implies the energy equation for the strong solution, namely
		$$
	\|u_2(t)\|_2^2 + \frac{2(1+K)}{\epsilon}\int_{0}^t\|\nabla u_2(s)\|_2^2ds = \|u_0\|_2^2 +
	\frac{4K}{\epsilon}\int_{0}^t ({\rm rot}\, \gamma_2(s),u_2(s))ds + \frac{2{\rm Ra}}{\epsilon} \int_{0}^t(\theta_2(s),(u_2)_3(s))ds.
	$$
	Testing \eqref{eq:strong_1} with $v(t)$ it follows that
\begin{align*}
	& \int_{0}^t \, ((u_2)_t(s), v(s)) \, ds = -\int_0^t((u_2(s)\cdot \nabla) u_2(s),v(s)){ds} - \frac{1+K}{\epsilon}\int_0^t(\nabla u_2(s), \nabla v(s))\, ds\\
	& \qquad \qquad  + \frac{2K}{\epsilon} \int_0^t({\rm rot}\, \gamma_2(s),v(s))\, ds + \frac{{\rm Ra}}{\epsilon} \int_0^t (\theta_2(s),v_3(s))\, ds.
	\end{align*}
	Finally testing \eqref{eq:strong_1} and \eqref{eq:1f_w} with $u_2(t)$ and subtracting the two relations it follows that
\begin{align*}
&\int_0^t\langle v_t(s),u_2(s)\rangle_{D(-P\Delta)'\times D(-P\Delta)} \, ds= - \int_0^t((u_1(t)\cdot \nabla) u_1(s),u_2(s))\, ds - \frac{(1+K)}{\epsilon}\int_0^t(\nabla v(s),\nabla u_2(s))\, ds \\
& \qquad \qquad + \frac{2K}{\epsilon} \int_0^t({\rm rot}\, \psi(s),u_2(s))\, ds + \frac{{\rm Ra}}{\epsilon} \int_0^t(\eta(s),(u_2)_3(s))\, ds. 
	\end{align*}
	Summarizing we deduce after some obvious transformations
\begin{align*}
&
	\|v(t)\|^2_{2} \leq - \frac{2(1+K)}{\epsilon}\int_{0}^t\|\nabla v(s)\|_2^2ds +
	\frac{4K}{\epsilon}\int_{0}^t ({\rm rot}\, \psi(s), v_3(s))ds\\
	& \qquad  + \frac{2{\rm Ra}}{\epsilon} \int_{0}^t(\eta(s),v_3(s))ds	
	  + 2\int_0^t((v(s)\cdot \nabla) v(s),u_2(s)){ds}. 
	\end{align*}
	But
\begin{align*}
&
	2\int_0^t((v(s)\cdot \nabla) v(s),u_2(s)){ds}\leq 2 \int_0^t\|v(s)\|_{L^2}\|\nabla v(s)\|_{L^2}\|u_2(s)\|_{L^\infty}{ds}\\
	&\qquad  \leq \delta \int_0^t \|\nabla v(s)\|^2_{L^2} \, ds + C(\delta) \int_0^t\|v(s)\|^2_{L^2}\|u_2(s)\|^2_{L^\infty}{ds},
	\end{align*}
	where the constant $\delta$ is arbitrary. Hence
	\begin{align*}
	&
	\|v(t)\|^2_{2} \leq C\left(\int_{0}^t \|\psi(s)\|_2^2ds + \int_{0}^t\|\eta(s)\|_2^2ds	
	+ \int_0^t\|v(s)\|^2_{2}\|u_2(s)\|^2_{L^\infty}{ds}\right). 
	\end{align*}
	Proceeding in a similar way with the remaining two equations we arrive at the bounds
	\begin{align*}
	&
	\|\psi(t)\|^2_{2} \leq C\left(\int_{0}^t \|v(s)\|_2^2ds 
	+ \int_0^t\|v(s)\|^2_{2}\|\gamma_2(s)\|^2_{L^\infty}{ds}\right),\\
	&\|\eta(t)\|^2_{2} \leq C\left(\int_{0}^t \|v(s)\|_2^2ds
	+ \int_0^t\|v(s)\|^2_{2}\|\theta_2(s)\|^2_{L^\infty}{ds}\right).
	\end{align*}
	Adding the three inequalities and denoting $F(t) = \|v(t)\|^2_{2} + \|\psi(t)\|^2_{2} + \|\eta(t)\|^2_{2}$ it follows that
	$$
	F(t) \leq C\left(\int_{0}^t F(s) ds +  \int_0^tF(s) \left( \|u_2(s)\|^2_{L^\infty}+\|\gamma_2(s)\|^2_{L^\infty}+ \|\theta_2(s)\|^2_{L^\infty}\right){ds}\right).
	$$
	Since, by the Agmon inequality
	$$
	\int_0^t\|u_2(s)\|^2_{L^\infty}\, ds \leq C\int_0^t\|\nabla u_2(s)\|_{2}\| u_2(s)\|_{D(-P\Delta)}\, ds \leq C \|u_2\|_{L^\infty(0,t;V)} \|u_2\|_{L^2(0,t;D(-P\Delta))},
	$$
	we obtain the assertion by the Gronwall lemma.
\end{proof}
%
%
%
%
%
\section{A priori estimates} \label{Basic estimates}

In this section we derive some estimates which will be used several times in the rest of the article. We introduce the new constant  $D  = \max\left\{ 2, \frac{M}{L}\right\}$ which will be useful in the estimates of this section.   The constants denoted by small $c$, such as $c_1, c_2, \ldots$ are universal and do not depend on the data of the problem. All dependence on the problem data is explicitly written in the estimates apart from constants $T_1, T_2, \ldots$, times needed to enter some absorbing balls. These times can depend on the constants present in the formulation of the problem, as well as on the initial condition and it is not always written explicitly. Note, however, that always times $T_1, T_2, \ldots$ are possible to be chosen uniformly for the initial data in the bounded sets $(u_0, \gamma_0, \theta_0) \in H\times E_3 \times E_1$.

\subsection{Maximum principle for temperature.} We start from Stampacchia type maximum principle estimates for temperature. Note that Wang \cite{Wang2007-Asymptotic} includes these estimates in his definition of ''suitable'' weak solutions. As we require our weak solution to be limits of approximative Galerkin problems, we derive these estimates as a consequence of our definition.  

Note that in the next lemma the bounds hold only for almost every $t>0$. This is the consequence of the fact that the limit of truncations of weakly convergent sequences does not have to be equal to the truncation of the limit. Hence, the estimate is established only in those time points in which the approximative sequences converge strongly.  
\begin{lemma}\label{lem_max_ae}
	Suppose that the triple $(u,\gamma,\theta)$ is the weak solution given by Definition \ref{def-weak-sol-evol}. Then for almost every $t > 0$ there hold the inequalities
	\begin{align}\label{ine-t-1:4f_w} 
	&\|(T-1)^+(t)\|_{L^2} \leq \|(T-1)^+(0)\|_{L^2}e^{-\lambda_1 t},\\
	&\label{ine-t-2:4f_w}
	\|T^-(t)\|_{L^2} \leq \|T^-(0)\|_{L^2}e^{-\lambda_1 t},
	\end{align}
	where $T(x_1,x_2,x_3,t)=\theta(x_1,x_2,x_3,t) + 1-x_3$.
\end{lemma}
\begin{proof}
	If $\theta^n$ satisfies \eqref{eq:4f_w_galer} then $T^n(x_1,x_2,x_3,t) = \theta^n(x_1,x_2,x_3,t) + 1-x_3$ satisfies
	$$
	\langle T^n_t(t), \eta \rangle + (u^n(t)\cdot \nabla T^n(t),\eta) + (\nabla  T^n(t),\nabla \eta)  = 0,
	$$
	for every $\eta\in W_1$ and a.e. $t\in (0,T)$. We can take $\eta = (T^n-1)^+$, whence
	$$
	\frac{1}{2}\frac{d}{dt}\|(T^n(t)-1)^+\|_{L^2}^2 + \|\nabla (T^n(t)-1)^+\|^2_{L^2} = 0.
	$$
	Integrating, it follows that
	$$
	\|(T^n(t)-1)^+\|_{L^2}^2 + 2\int_{0}^t\|\nabla (T^n(s)-1)^+\|^2_{L^2}\, ds = \|(T^n(0)-1)^+\|_{L^2}^2.
	$$
	The Poincar\'{e} inequality implies that
		$$
	\|(T^n(t)-1)^+\|_{L^2}^2 + 2\lambda_1\int_{t_0}^t\|\nabla (T^n(s)-1)^+\|^2_{L^2}\, ds \leq  \|(T^n(0)-1)^+\|_{L^2}^2.
	$$
	Now, the Gronwall lemma implies that
	 $$
	 \|(T^n(t)-1)^+\|_{L^2} \leq \|(T^n(0)-1)^+\|_{L^2} e^{-\lambda_1 t}.
	 $$
	The convergence $\|(T^n(t)-1)^+\|_{L^2} \to \|(T(t)-1)^+\|_{L^2}$ holds for a.e. $t > 0$ and for $t=0$, cf. Lemma \ref{lemma:convergence} and hence we can pass to the limit in those time points to get \eqref{ine-t-1:4f_w}. Assertion \eqref{ine-t-2:4f_w} follows by taking $(T^n)^-$ as the test function in place of $(T^n-1)^+$.  
\end{proof}

\begin{lemma}\label{lem:max_temp}
	If $(u,\gamma,\theta)$ is a weak solution of system \eqref{eq:1f}--\eqref{eq:4f} in the sense of Definition~\ref{def-weak-sol-evol} 
then
	\begin{equation}\label{eq:max_theta}
		\|\theta(t)\|_{2} \leq \left(\frac{\sqrt{3}}{3}+1\right)\sqrt{A} + 2\|T_0\|_{2} e^{-\lambda_1 t} \leq 2\sqrt{A} + 2\left(\|\theta_0\|_{2}+\sqrt{A}\right) e^{-\lambda_1 t}.
	\end{equation}	
	where $A = L_{x_1}/h\cdot L_{x_2}/h$.
In particular $\|T(t)\|_{2}$ and $\|\theta(t)\|_{2}$ are bounded uniformly in time by the quantity which depends only on the initial data $\theta_0$ and the geometry of the domain 
$\Omega$.
\end{lemma}
\begin{proof}
We decompose $T(t) = T_1(t) + T_2(t)$, where $T_1(t) = (T-1)^+(t) - T^-(t)$ and $T_2(t)$ has values in $[0,1]$ for a.e. $x\in \Omega$. So, 
$$
\theta(t) = -(1-x_3) + (T-1)^+(t) - T^-(t) + T_2(t).
$$
It follows that 
$$
\|\theta(t)\|_2 \leq \|T_2(t)\|_2 + \|(1-x_3)\|_2 + \|(T-1)^+(t)\|_2 + \|T^-(t)\|_2 \leq \left(\frac{\sqrt{3}}{3}+1\right)\sqrt{A} + \|(T-1)^+(t)\|_2 + \|T^-(t)\|_2.
$$
Hence, by Lemma \ref{lem_max_ae} for almost every $t > 0$ it holds
$$
\|\theta(t)\|_2 \leq \left(\frac{\sqrt{3}}{3}+1\right)\sqrt{A} + \left(\|(T-1)^+(0)\|_{L^2} + \|T^-(0)\|_{L^2}\right)e^{-\lambda_1 t} \leq \left(\frac{\sqrt{3}}{3}+1\right)\sqrt{A} + 2\|T_0\|_{L^2}e^{-\lambda_1 t}.
$$
Now, the fact that $\theta \in C_w([0,\infty);H)$ together with the sequential weak lower semicontinuity of the norm implies that this inequality holds actually for every $t\geq 0$. 
\end{proof}
%


\subsection{Energy estimates.} In this section we derive the estimates which follow from the energy relations. 

\begin{lemma}\label{le:inf-est-1}
If $(u,\gamma,\theta)$ is a weak solution of system \eqref{eq:1f}--\eqref{eq:4f} in the sense of Definition~\ref{def-weak-sol-evol} then the following bounds hold for every $t\geq 0$: if $D\epsilon \neq 1$
\begin{align}\label{eq:u-omega-l2-est}
	\|u(t)\|_2^2 + M\|\gamma(t)\|_2^2 \leq  \left(\|u_0\|^2 + M\|\gamma_0\|_2^2\right)e^{-\frac{2\lambda_1}{D\epsilon}t}  + \frac{8D{\rm Ra}^2(\|\theta_0\|_2+\sqrt{A})^2}{|D\epsilon-1|} e^{-2\lambda_1\min\left\{1,\frac{1}{D\epsilon}\right\}t} +  8AD{\rm Ra}^2,
\end{align}
and if $D\epsilon = 1$ then
\begin{align}\label{eq:u-omega-l2-est-2}
\|u(t)\|_2^2 + M\|\gamma(t)\|_2^2 \leq \left(\|u_0\|^2 + M\|\gamma_0\|_2^2+\frac{16{\rm Ra}^2 (\|\theta_0\|_2+\sqrt{A})^2}{\epsilon} \right)   e^{-\lambda_1 t} + 8AD{\rm Ra}^2.
\end{align}
\end{lemma}
\begin{proof} 
	Adding \eqref{ine-u:4f_w} and \eqref{ine-omega:4f_w} it follows that
\begin{align*}
& \|u(t)\|_2^2 + M\|\gamma(t)\|_2^2 + \frac{2L}{\epsilon}\int_{t_0}^t\|\nabla \gamma(s)\|_2^2ds +  \frac{2}{\epsilon}\int_{t_0}^t\|\nabla u(s)\|_2^2ds \\
&\quad  + \frac{2K}{\epsilon}\int_{t_0}^t \|2\gamma(s)-{\rm rot}\, u(s)\|_2^2\, ds  \leq \frac{2{\rm Ra}}{\epsilon} \int_{t_0}^t(\theta(s),u_3(s))ds + \|u(t_0)\|_2^2+M\|\gamma(t_0)\|_2^2.
	\end{align*}
	Dropping the term with $K$ and using the Poincar\'{e} inequality we obtain
	\begin{align*}
	& \|u(t)\|_2^2 + M\|\gamma(t)\|_2^2 + \frac{2L}{\epsilon}\int_{t_0}^t\|\nabla \gamma(s)\|_2^2ds +  \frac{2}{\epsilon}\int_{t_0}^t\|\nabla u(s)\|_2^2ds \\
	&\quad   \leq \frac{2{\rm Ra}\sqrt{\lambda_1}}{\epsilon} \int_{t_0}^t \|\theta(s)\|_2 \|\nabla u(t)\| ds + \|u(t_0)\|_2^2+M\|\gamma(t_0)\|_2^2.
	\end{align*}
	We easily deduce
	\begin{align*}
	& \|u(t)\|_2^2 + M\|\gamma(t)\|_2^2 + \frac{2L}{\epsilon}\int_{t_0}^t\|\nabla \gamma(s)\|_2^2ds +  \frac{1}{\epsilon}\int_{t_0}^t\|\nabla u(s)\|_2^2ds \\
	&\quad   \leq \frac{2{\rm Ra}^2 \lambda_1}{\epsilon} \int_{t_0}^t \|\theta(s)\|^2_2  ds + \|u(t_0)\|_2^2+M\|\gamma(t_0)\|_2^2.
	\end{align*}
	The Poincar\'{e} inequality now implies
	\begin{align*}
	& \|u(t)\|_2^2 + M\|\gamma(t)\|_2^2 + \frac{2L\lambda_1}{\epsilon}\int_{t_0}^t\|\gamma(s)\|_2^2ds +  \frac{\lambda_1}{\epsilon}\int_{t_0}^t\|u(s)\|_2^2ds \\
	&\quad   \leq \frac{2{\rm Ra}^2 \lambda_1}{\epsilon} \int_{t_0}^t \|\theta(s)\|^2_2  ds + \|u(t_0)\|_2^2+M\|\gamma(t_0)\|_2^2.
	\end{align*}
	Using the previously defined constant $D$, we obtain 
		\begin{align*}
	& \|u(t)\|_2^2 + M\|\gamma(t)\|_2^2 + \frac{2\lambda_1}{D\epsilon}\int_{t_0}^t\|u(s)\|_2^2 + M\|\gamma(s)\|_2^2ds \\
	&\qquad  \leq \frac{8{\rm Ra}^2 \lambda_1}{\epsilon} \int_{t_0}^t \left(\sqrt{A} + \left(\|\theta_0\|_{2}+\sqrt{A}\right) e^{-\lambda_1 s}\right)^2  ds + \|u(t_0)\|_2^2+M\|\gamma(t_0)\|_2^2.
	\end{align*}
	A direct computation leads to
	\begin{align*}
	& \|u(t)\|_2^2 + M\|\gamma(t)\|_2^2 + \frac{2\lambda_1}{D\epsilon}\int_{t_0}^t\|u(s)\|_2^2 + M\|\gamma(s)\|_2^2ds \\
	&\qquad  \leq \frac{8{\rm Ra}^2 \lambda_1}{\epsilon} \left(2A(t-t_0)+\frac{(\|\theta_0\|_2+\sqrt{A})^2}{\lambda_1}(e^{-2\lambda_1 t_0}-e^{-2\lambda_1 t})\right) + \|u(t_0)\|_2^2+M\|\gamma(t_0)\|_2^2.
	\end{align*}
	Suppose that $D\epsilon \neq 1$. Define 
	$$
	f(t) =  -8{\rm Ra}^2D \left(\frac{(\|\theta_0\|_2+\sqrt{A})^2}{1-D\epsilon}e^{-2\lambda_1t} + A\right).
	$$ 
	Then
	\begin{align}
& \|u(t)\|_2^2 + M\|\gamma(t)\|_2^2 + f(t) + \frac{2\lambda_1}{D\epsilon}\int_{t_0}^t\|u(s)\|_2^2 + M\|\gamma(s)\|_2^2 + f(s) ds \label{est:imp}\\
&\qquad  \leq f(t_0) + \|u(t_0)\|_2^2+M\|\gamma(t_0)\|_2^2.\nonumber
\end{align}
	Using \cite[Lemma 7.2]{Ball} it follows that
	$$
	\|u(t)\|_2^2 + M\|\gamma(t)\|_2^2 + f(t) \leq e^{-\frac{2\lambda_1}{D\epsilon}t} (\|u_0\|^2 + M\|\gamma_0\|_2^2 + f(0)).
	$$
	Hence if $D\epsilon < 1$
	$$
	\|u(t)\|_2^2 + M\|\gamma(t)\|_2^2 \leq  \left(\|u_0\|^2 + M\|\gamma_0\|_2^2 \right)e^{-\frac{2\lambda_1}{D\epsilon}t} + \frac{8D{\rm Ra}^2(\|\theta_0\|_2+\sqrt{A})^2}{1-D\epsilon}e^{-2\lambda_1t} + 8AD{\rm Ra}^2 .
	$$
	If $D\epsilon > 1$
		$$
	\|u(t)\|_2^2 + M\|\gamma(t)\|_2^2 \leq  \left(\|u_0\|^2 + M\|\gamma_0\|_2^2 + \frac{8D{\rm Ra}^2(\|\theta_0\|_2+\sqrt{A})^2}{D\epsilon-1} \right)e^{-\frac{2\lambda_1}{D\epsilon}t} +  8AD{\rm Ra}^2 .
	$$
		and the assertion for $D\epsilon \neq 1$ is proved. It remains to verify the case $D\epsilon = 1$. Then
		\begin{align*}
	& \|u(t)\|_2^2 + M\|\gamma(t)\|_2^2 + 2\lambda_1\int_{t_0}^t\|u(s)\|_2^2 + M\|\gamma(s)\|_2^2ds \\
	&\qquad  \leq \frac{8{\rm Ra}^2 \lambda_1}{\epsilon} \left(2A(t-t_0)+\frac{(\|\theta_0\|_2+\sqrt{A})^2}{\lambda_1}(e^{-2\lambda_1 t_0}-e^{-2\lambda_1 t})\right) + \|u(t_0)\|_2^2+M\|\gamma(t_0)\|_2^2,
	\end{align*}
To get the assertion we should take 
	$$
	f(t) = -8{\rm Ra}^2D \left((\|\theta_0\|_2+\sqrt{A})^2 2\lambda_1 t e^{-2\lambda_1 t} + A\right),
	$$
	in order to obtain \eqref{est:imp}. Again, we can use \cite[Lemma 7.2]{Ball} which leads us to
		$$
	\|u(t)\|_2^2 + M\|\gamma(t)\|_2^2 \leq e^{-\frac{2\lambda_1}{D\epsilon}t} (\|u_0\|^2 + M\|\gamma_0\|_2^2) + \frac{8{\rm Ra}^2 (\|\theta_0\|_2+\sqrt{A})^2}{\epsilon}  2\lambda_1 t  e^{-2\lambda_1 t} + 8AD{\rm Ra}^2.
	$$
	As $\lambda_1t e^{-\lambda_1t} \leq e^{-1}$ the assertion follows easily.
\end{proof}

Denote  $V^n(t) = M\|\nabla \gamma^n(t)\|_2^2 + \|\nabla u^n(t)\|^2_2$. 

\begin{lemma}\label{lemma:t_1_integral}
If 
$$T_1 = \frac{1}{\pi^2}\ln\left(1+\frac{\|\theta_0\|_2}{\sqrt{A}}\right),
$$
and $t > T_1$, then
\begin{align}
&\frac {1}{t-T_1}\int_{T_1}^{t} V^n(s)ds  \leq \frac{\epsilon}{2D(t-T_1)} ( \|u^n(t_1)\|_2^2 +  M\|\gamma^n(t_1)\|_2^2 ) + \frac{8 A{\rm Ra}^2 }{D \pi^2}.\label{int:ass1_t1}
\end{align}
\end{lemma}
\begin{proof}
Adding the estimates \eqref{ine-u:4f_w} and \eqref{ine-omega:4f_w} we obtain
\begin{align*}
& \epsilon\|u^n(t)\|_2^2 + \epsilon M\|\gamma^n(t)\|_2^2 +2L\int_{T_1}^{t}\|\nabla \gamma^n(s)\|_2^2ds + 2G\int_{T_1}^{t}\|\mathrm{div}\, \gamma^n(s)\|_2^2ds \\
&\qquad \qquad + 2K\int_{T_1}^{t} \|2\gamma^n(s)-{\rm rot}\, u^n(s)\|_2^2\, ds + 2\int_{T_1}^{t}\|\nabla u^n(s)\|_2^2ds\\ 
&\qquad \leq \epsilon\|u^n(T_1)\|_2^2  + 2{\rm Ra} \int_{T_1}^{t}(\theta^n(s),u^n_3(s))ds+\epsilon M\|\gamma^n(T_1)\|_2^2.
\end{align*}
Since, for $t\geq T_1$ there holds $\|\theta^n(t)\|\leq 4\sqrt{A}$, we deduce
\begin{align*}
& 2L\int_{T_1}^{t}\|\nabla \gamma^n(s)\|_2^2ds +  2\int_{T_1}^{t}\|\nabla u^n(s)\|_2^2ds \leq \epsilon\|u^n(T_1)\|_2^2  + 8\sqrt{A}{\rm Ra} \int_{T_1}^{t}\|u^n(s)\|_2ds+\epsilon M\|\gamma^n(T_1)\|_2^2.
\end{align*}
The Poincar\'{e} inequality implies 
\begin{align*}
& 2\left(L\int_{T_1}^{t}\|\nabla \gamma^n(s)\|_2^2ds +  \int_{T_1}^{t}\|\nabla u^n(s)\|_2^2ds\right) \leq \epsilon\|u^n(T_1)\|_2^2  + \frac{8\sqrt{A}{\rm Ra}}{\pi} \int_{T_1}^{t}\|\nabla u^n(s)\|_2ds+\epsilon M\|\gamma^n(T_1)\|_2^2.
\end{align*}
The Cauchy inequality implies that
\begin{align*}
& 2\left(L\int_{T_1}^{t}\|\nabla \gamma^n(s)\|_2^2ds +  \frac{1}{2}\int_{T_1}^{t}\|\nabla u^n(s)\|_2^2ds\right) \leq \epsilon\|u^n(T_1)\|_2^2  + \frac{16 A{\rm Ra}^2 }{ \pi^2}(t-T_1) + \epsilon M\|\gamma^n(T_1)\|_2^2.
\end{align*}
We deduce
$$  M\int_{T_1}^{t}\|\nabla \gamma^n(s)\|_2^2ds +  \int_{T_1}^{t}\|\nabla u^n(s)\|_2^2ds  \leq \frac{\epsilon}{2D} ( \|u^n(T_1)\|_2^2 +  M\|\gamma^n(T_1)\|_2^2 ) + \frac{8 A{\rm Ra}^2 }{D \pi^2} (t-T_1),
$$
whence the assertion follows.
\end{proof}

\begin{remark}
	Estimates of Lemmas \ref{le:inf-est-1},  
	 and \ref{lemma:t_1_integral}  hold also for the solutions of the approximative problems given by Definition \ref{def-weak-sol-evol-galerkin}  with the constants independent on the dimension of the finite dimensional spaces that approximate $u$ and $\gamma$. 
\end{remark}

\subsection{Enstrophy estimates.} The crucial assumption for our results below is that 
the constants present in the problem satisfy the following restrictions meaning that the Prandtl number ${\rm Pr}$ and the micropolar damping $L$ are large enough. 
\begin{itemize}
	\item[(H)] $L \geq \frac{16 }{3\pi^2} K
	\quad \textrm{and}\quad 
	{\rm Pr} \geq 2{\rm Ra}{D^{3/2}c_1\sqrt{A}}.$
\end{itemize}

In the following lemma which follows from the enstrophy estimates we show that if restrictions (H) hold, then there exists a ball in $H^1$ that is forward invariant for large time for the approximative problems. 
\begin{lemma}\label{lem:enstrophy}
Assume (H). Let $(u^n,\gamma^n,\theta^n)$ be the solution to the Galerkin problem given in Definition \ref{def-weak-sol-evol-galerkin}. 
If 
$$T_1 = \frac{1}{\pi^2}\ln\left(1+\frac{\|\theta_0\|_2}{\sqrt{A}}\right),
$$
and for some $t\geq T_1$ there holds $(u^n(t),\gamma^n(t)) \in S_R$ then for every $s\geq t$ there also holds $(u^n(s),\gamma^n(s)) \in S_R$, where 
$
S_R = \{ (u,\gamma)\in V\times W_3  \,:\  M\|\nabla \gamma\|_2^2 + \|\nabla u\|^2_2 \leq R \}
$ and 
$$
R = \frac{4\pi}{\sqrt{6}}AD{\rm Ra}^2.
$$
\end{lemma}
\begin{proof}
We test the Galerkin equation \eqref{eq:1f_w_galer} by $-P\Delta u^n(t)$, the value of the Stokes operator applied to $u^n(t)$, which gives
\begin{align*}
 & \frac{\epsilon}{2}\frac{d}{dt}\|\nabla u^n(t)\|^2_2+  (1+K)\|P\Delta u^n(t)\|_2^2 \\
 & \qquad \leq - 2K ({\rm rot}\, \gamma^n(t),P\Delta u^n(t)) - {\rm Ra} (\theta^n(t),(P\Delta u^n(t))_3) - \epsilon((u^n(t)\cdot \nabla) u^n(t),(P\Delta u^n(t))).
\end{align*}
We deduce, using the Cauchy, and Agmon inequalities
\begin{align}
& \frac{\epsilon}{2}\frac{d}{dt}\|\nabla u^n(t)\|^2_2+  (1+K)\|P\Delta u^n(t)\|_2^2\nonumber \\
& \qquad \leq  2K\|\nabla \gamma^n(t)\|_2\|P\Delta u^n(t)\|_2 + {\rm Ra} \|\theta^n(t)\|_2\|P\Delta u^n(t)\|_2 + c_1\epsilon\|\nabla u^n(t)\|^{3/2}_2\|P\Delta u^n(t)\|^{3/2}_2\label{eq:enstrophy_1},
\end{align}
where $c_1$ is the constant from the Agmon inequality $\|v\|_{L^\infty}\leq c_1 \|\nabla v\|_2^{1/2}
\|P\Delta u\|_2^{1/2}$. 
Observe that
\begin{align}\label{3D-case}
(\nabla{\rm div\,}\gamma^n(t), \Delta\gamma^n(t)) = \|\nabla{\rm div\,}\gamma^n(t)\|_2^2
\end{align}
since $\Delta \gamma^n(t) = \nabla{\rm div\,}\gamma^n(t) - {\rm rot}\,{\rm rot}\,\gamma^n(t)$, ${\rm rot}\,\nabla F =0$, and 
$(\nabla F, {\rm rot}\,{\rm rot}\,\gamma^n(t))=({\rm rot}\,
\nabla F, {\rm rot}\,\gamma^n(t)) =0$.
Hence, testing \eqref{eq:3f} with $-\Delta \gamma^n(t)$ we can drop the term with $G$ and we deduce that
\begin{align*}
& 	\frac{M\epsilon}{2}\frac{d}{dt}\|\nabla \gamma^n(t)\|_2^2  + L\|\Delta \gamma^n(t)\|_2^2  +
4 K \|\nabla \gamma^n(t)\|_2^2 \\
& \qquad \leq  2K \|\nabla u^n(t)\|_2\|\Delta \gamma^n(t)\|_2 + c_1 \epsilon M \|\nabla u^n(t)\|_2^{1/2} \|P\Delta  u^n(t)\|_2^{1/2} \|\nabla\gamma^n(t)\|_2 \|\Delta \gamma^n(t)\|_2
\end{align*}
After some simple calculations it follows that
\begin{align*}
& 	\frac{M\epsilon}{2}\frac{d}{dt}\|\nabla \gamma^n(t)\|_2^2  + \frac{L}{2}\|\Delta \gamma^n(t)\|_2^2  +
4 K \|\nabla \gamma^n(t)\|_2^2 \\
& \qquad \leq  \frac{4K^2}{L} \|\nabla u^n(t)\|^2_2 + \frac{c_1^2 \epsilon^2 M^2}{L} \|\nabla u^n(t)\|_2 \|P\Delta  u^n(t)\|_2 \|\nabla\gamma^n(t)\|^2_2.
\end{align*}
We add the resulting inequality to \eqref{eq:enstrophy_1}, and use the Cauchy inequality once again which yields
\begin{align*}
& 	\frac{M\epsilon}{2}\frac{d}{dt}\|\nabla \gamma^n(t)\|_2^2 + \frac{\epsilon}{2}\frac{d}{dt}\|\nabla u^n(t)\|^2_2+ \frac{L}{2}\|\Delta \gamma^n(t)\|_2^2 +  \left(1+\frac{3K}{4}\right)\|P\Delta u^n(t)\|_2^2\\
& \qquad \leq  \frac{4K^2}{L} \|\nabla u^n(t)\|^2_2  + {\rm Ra} \|\theta^n(t)\|_2\|P\Delta u^n(t)\|_2 \\
& \qquad \qquad + \frac{c_1^2 \epsilon^2 M^2}{L} \|\nabla u^n(t)\|_2 \|P\Delta  u^n(t)\|_2 \|\nabla\gamma^n(t)\|^2_2 + c_1\epsilon\|\nabla u^n(t)\|^{3/2}_2\|P\Delta u^n(t)\|^{3/2}_2.
\end{align*}
After further computations, which use the Poincar\'{e} inequality, the relation between $K$ and $L$ in (H) and the fact that $t\geq T_1$ we obtain
\begin{align*}
& 	\frac{\epsilon}{2}\frac{d}{dt}V^n(t) + \frac{\pi^2L}{2}\|\nabla \gamma^n(t)\|_2^2 +  \|P\Delta u^n(t)\|_2^2\\
& \qquad \leq 4{\rm Ra} \sqrt{A}\|P\Delta u^n(t)\|_2 + \frac{c_1^2 \epsilon^2 M^2}{L} \|\nabla u^n(t)\|_2 \|P\Delta  u^n(t)\|_2 \|\nabla\gamma^n(t)\|^2_2 + c_1\epsilon\|\nabla u^n(t)\|^{3/2}_2\|P\Delta u^n(t)\|^{3/2}_2.
\end{align*}
Using the Young inequality  we deduce that for any $\delta>0$ there holds
\begin{align*}
& 	\frac{\epsilon}{2}\frac{d}{dt}V^n(t) + \frac{\pi^2L}{2}\|\nabla \gamma^n(t)\|_2^2 +  \|P\Delta u^n(t)\|_2^2\\
& \qquad \leq \frac{3\delta}{2}\|P\Delta u^n(t)\|_2^2 + \frac{8 {\rm Ra} ^2 A}{\delta} + \frac{c_1^4 \epsilon^4 M^4}{2\delta L^2} \|\nabla u^n(t)\|^2_2 \|\nabla\gamma^n(t)\|^4_2 + \frac{3^3 c_1^4\epsilon^4}{2^5\delta^3} \|\nabla u^n(t)\|^6_2.
\end{align*}
Assuming that $\delta < 2/3$ and using the Poincar\'{e} inequality, this yields 
\begin{align*}
& 	\epsilon\frac{d}{dt}V^n(t) + \pi^2L\|\nabla \gamma^n(t)\|_2^2 +  \pi^2(2-3\delta)\|\nabla u^n(t)\|_2^2\\
& \qquad \leq \frac{16 {\rm Ra} ^2 A}{\delta} + \frac{c_1^4 \epsilon^4 M^4}{\delta L^2} \|\nabla u^n(t)\|^2_2 \|\nabla\gamma^n(t)\|^4_2 + \frac{3^3 c_1^4\epsilon^4}{2^4\delta^3} \|\nabla u^n(t)\|^6_2.
\end{align*}
After some simple calculations we arrive at
\begin{align*}
& 	\epsilon\frac{d}{dt}V^n(t) + \pi^2L\|\nabla \gamma^n(t)\|_2^2 +  \pi^2(2-3\delta)\|\nabla u^n(t)\|_2^2\\
& \qquad \leq \frac{16 {\rm Ra} ^2 A}{\delta} + c_1^4 \epsilon^4  \left(\frac{2+3^3}{2^4\delta^3}\|\nabla u^n(t)\|^6_2+\frac{2M^6}{L^3}\|\nabla\gamma^n(t)\|^6_2\right).
\end{align*}
where $\theta > 0$ is arbitrary. Setting $\delta = 1/2$, and using the notation $D = \max\left\{2,\frac{M}{L}\right\}$, we obtain
\begin{align*}
& 	\epsilon\frac{d}{dt}V^n(t) + \frac{\pi^2}{D} V^n(t)\\
& \qquad \leq 32 {\rm Ra} ^2 A + c_1^4 \epsilon^4  2\max\left\{ \frac{2+ 3^3}{2^2}, \frac{M^3}{L^3}\right\} V^n(t)^3.
\end{align*}
It follows that
\begin{align*}
& 	\epsilon\frac{d}{dt}V^n(t) \leq 32 {\rm Ra} ^2 A + 2c_1^4 \epsilon^4  D^3V^n(t)^3 - \frac{\pi^2}{D} V^n(t).
\end{align*}
We need to find sufficiently large $R$ such that the set $S_R$ is forward invariant for  $t\geq T_1$. 
It is sufficient to find $R>0$ such that
\begin{equation}\label{eq:radius}
32 {\rm Ra} ^2 A + 2 c_1^4 \epsilon^4  D^3V^n(t)^3 - \frac{\pi^2}{D} V^n(t) \leq 0.
\end{equation}
If only such $R$ exists, then $S_R$ is forward invariant for $t\geq T_1$. It is a straightforward and cumbersome computation to check that if only
$$
\frac{\rm Ra}{\rm Pr} = \epsilon {\rm Ra} \leq \frac{1}{2} \frac{1}{D^{3/2}c_1\sqrt{A}},
$$
then we can choose 
$$
R = \frac{4\pi}{\sqrt{6}}AD{\rm Ra}^2,
$$
and the inequality \eqref{eq:radius} is satisfied. The proof is complete. 
\end{proof}
\begin{lemma}\label{lem:invariance}
Assume (H). 
	If the initial data $(u_0,\gamma_0,
	\theta_0)\in B$, where $B$ is a bounded set in $H\times E_3\times E_1$ then there exists $T_2>0$ dependent on $B$ such that 
	$$
	\|\nabla u^n(t)\|_2^2 + M\|\nabla \gamma^n(t)\|_2^2 \leq  \frac{4\pi}{\sqrt{6}}AD{\rm Ra}^2\quad \textrm{for every}\quad t\geq T_2. 
	$$ 
\end{lemma}
\begin{proof}
	Lemmas \ref{lem:max_temp} and \ref{le:inf-est-1} imply that 
	$$
	\|u^n(t)\|_2^2 + \|\gamma^n(t)\|_2 + \|\theta^n(t)\|_2^2 \leq C \quad \textrm{for every}\quad t\geq 0, 
	$$
	where the constant $C$ depends on $B$. Let $T_1$ be as in Lemma \ref{lem:enstrophy}. Thus there exists time $T_2 = T_2(B) > 0$ such that 
	$$
	\frac{\epsilon}{2D(T_2-T_1)} ( \|u^n(T_1)\|_2^2 +  M\|\gamma^n(T_1)\|_2^2 ) \leq \frac{8 A{\rm Ra}^2 }{D \pi^2}.
	$$
	By Lemma \ref{lemma:t_1_integral} this means that
	$$
	\frac {1}{T_2-T_1}\int_{T_1}^{t} V^n(s)ds  \leq \frac{16 A{\rm Ra}^2 }{D \pi^2}.
	$$ 
	Since the function $V^n$ is continuous we deduce that there exists $t\in [T_1,T_2]$ such that
	$$
	V^n(t) \leq \frac{16 A{\rm Ra}^2 }{D \pi^2}.
	$$
	If only 
	$$
	\frac{16 A{\rm Ra}^2 }{D \pi^2} \leq \frac{4\pi}{\sqrt{6}}AD{\rm Ra}^2,
	$$
	then result follows by Lemma \ref{lem:enstrophy}. But this inequality is equivalent to 
	$$
	D^2 \geq \frac{4\sqrt{6}}{\pi^3},
	$$
	which is always true as $D\geq 2$, and the result follows. 
\end{proof}

In the next result we show that the absorbing ball exists not only for the approximate solutions but also for the limit solutions.
\begin{lemma}\label{lem:enstrophy_2}
	Assume (H). Let $B$ be a bounded set in $H\times E_3\times E_1$. There exists a time $T_2$ dependent only on $B$ such that for any weak solution $(u,\gamma,\theta)$ with the initial data in $B$ there holds 
	$$
	\|\nabla u(t)\|_2^2 + M\|\nabla \gamma(t)\|_2^2 \leq  \frac{4\pi}{\sqrt{6}}AD{\rm Ra}^2\quad \textrm{for every}\quad t\geq T_2. 
	$$
\end{lemma}
\begin{proof}
	If $(u,\gamma,\theta)$ is a weak solution, then there exists a certain sequence of indexes $n_k$ such that convergences of Lemma \ref{lemma:convergence} hold for this sequence. Fix $t\geq T_2$, where $T_2$ is the same as in Lemma \ref{lem:invariance}. Then by Lemma \ref{lem:invariance} the expression $\|\nabla u^{n_k}(t)\|_2^2 + \|\nabla \gamma^{n_k}(t)\|_2^2$ is bounded uniformly with respect to $n$, whence, for a subsequence, denoted again by $n_k$
	$$
	 u^{n_k}(t) \to a\quad \textrm{weakly in}\quad V\quad \textrm{and}\quad \gamma^{n_k}(t)\to b\quad \textrm{weakly in}\quad W_3\quad \textrm{as}\quad n\to \infty.  
	$$
	In view of Lemma \ref{lemma:convergence} $a=u(t)$ and $b=\gamma(t)$ and the convergences hold for the whole subsequence $n_k$. Sequential weak lower semicontinuity of norms implies that
	\begin{align*}
	& \|\nabla u(t)\|_2^2 + \|\nabla \gamma(t)\|_2^2 \leq \liminf_{k\to\infty} \|\nabla u^{n_k}(t)\|_2^2 + \liminf_{k\to\infty} \|\nabla \gamma^{n_k}(t)\|_2^2\\
	& \qquad \leq  \liminf_{k\to\infty} (\|\nabla u^{n_k}(t)\|_2^2 + \|\nabla \gamma^{n_k}(t)\|_2^2) \leq \frac{4\pi}{\sqrt{6}}AD{\rm Ra}^2,
	\end{align*}
	and the proof is complete. 
\end{proof}
\subsection{Gradient estimates for temperature.} We pass to the derivation of the uniform estimates for $\|\nabla \theta\|_2$. Before we derive the estimate for this value, we need an auxiliary estimate for the time integral of temperature for the approximative problem.   
\begin{lemma}\label{lemma:Temp_int}
	Assume (H). Let $u_n, \gamma_n, \theta_n$ be the solution of the approximative problem given in Definition \ref{def-weak-sol-evol-galerkin} where the initial data $(u_0, \gamma_0, \theta_0)$ belongs to a bounded set in $H\times E_3\times E_1$. There exists a constant $T_3$ depending on $B$ and the constants present in the formulation of the problem such that for every $t_2 > t_1 \geq T_3$ there holds
	$$e^{-\pi^2 t_2}\int_{t_1}^{t_2} e^{\pi^2 s} \|\nabla  \theta^n(s)\|_2^2\, ds \leq 
	16 A e^{-\pi^2 (t_2-t_1)} + \frac{64 A\sqrt{D}{\rm Ra}}{\pi^2}.
	$$
\end{lemma}
\begin{proof}
	Testing \eqref{eq:4f_w_galer} with $\theta^n(t)$ it follows that 
	$$
	\frac{1}{2}\frac{d}{dt} \|\theta^n(t)\|_2^2 + \|\nabla  \theta^n(t)\|_2^2  = (u^n_3(t),\theta(t)).
	$$
	We deduce
	$$
	\frac{d}{dt} \|\theta^n(t)\|_2^2 + \|\nabla  \theta^n(t)\|_2^2 + \lambda_1\|\theta^n(t)\|_2^2 \leq  2\|u^n(t)\|\|\theta^n(t)\|.
	$$
	Multiplying by the integrating factor $e^{\lambda_1 t}$ and integrating from $t_1$ to $t_2$ it follows that
	$$\|\theta^n(t_2)\|_2^2 e^{\lambda_1 t_2} +\int_{t_1}^{t_2} e^{\lambda_1 s} \|\nabla  \theta^n(s)\|_2^2\, ds \leq 
	\|\theta^n(t_1)\|_2^2 e^{\lambda_1 t_1} + 2\int_{t_1}^{t_2}e^{\lambda_1 s}\|u^n(s)\|\|\theta^n(s)\|\, ds.
	$$ 
	Consequently
	$$e^{-\lambda_1 t_2}\int_{t_1}^{t_2} e^{\lambda_1 s} \|\nabla  \theta^n(s)\|_2^2\, ds \leq 
	\|\theta^n(t_1)\|_2^2 e^{-\lambda_1 (t_2-t_1)} + e^{-\lambda_1 t_2}2\int_{t_1}^{t_2}e^{\lambda_1 s}\|u^n(s)\|\|\theta^n(s)\|\, ds.
	$$
	Due to estimates \eqref{eq:u-omega-l2-est} and \eqref{eq:max_theta} which are valid also for the approximative solution we can choose $T_3$ sufficiently large (depending on the initial data $u_0, \gamma_0, \theta_0$ and constants $A, {\rm Ra}, \epsilon, M, L, \lambda_1$) such that 
	$$
	\|u^n(t)\|_2 \leq 4\sqrt{AD}{\rm Ra}\quad \textrm{and}\quad \|\theta^n(t)\|_2 \leq 4\sqrt{A} \quad \textrm{for every}\quad t\geq T_3.
	$$
    So, if only $t_2 > t_1 \geq T_3$, then
	$$e^{-\lambda_1 t_2}\int_{t_1}^{t_2} e^{\lambda_1 s} \|\nabla  \theta^n(s)\|_2^2\, ds \leq 
	16 A e^{-\lambda_1 (t_2-t_1)} + \frac{64 A\sqrt{D}{\rm Ra}}{\lambda_1}.
	$$
	The proof is complete.
\end{proof}

\begin{lemma}\label{le:grad-temp-est-H1}
Assume (H).
 Let $B$ be a bounded set of initial data in $W\times E_3\times E_1$ and	let $u_n, \gamma_n, \theta_n$ be the solution of the approximative problem given in Definition \ref{def-weak-sol-evol-galerkin} and $u, \gamma, \theta$ be a weak solution given in Definition \ref{def-weak-sol-evol}. There exists a constant $T_4$ depending on the constants of the problem and the set $B$ and a universal constant $c_4$ such that for every $t \geq T_4$ there hold the estimates
\begin{align}\label{eq:grad-temp-est-H1}
&\|\nabla \theta^n(t)\|_2^2 \leq c_4 A D^{3/2}(1+A)(1+\mathrm{Ra}^3),\\
& \|\nabla \theta(t)\|_2^2 \leq c_4 A D^{3/2}(1+A)(1+\mathrm{Ra}^3).\label{eq:est-grad-temp-weak}
\end{align}
\begin{proof}
Testing \eqref{eq:4f_w_galer}  by $-\Delta\theta^n(t)$  we obtain
$$
\frac{1}{2}\frac{d}{dt}\|\nabla \theta^n(t)\|_2^2   + \|\Delta \theta^n(t)\|_2^2 \leq \|u^n_3(t)\|_2\|\Delta \theta^n(t)\|_2 + \|u^n(t)\|_6\|\|\nabla\theta^n(t)\|_{3}\|\Delta \theta^n(t)\|_2. 
$$
By the Gagliardo--Nirenberg inequality, and the continuity of embedding $V\subset L^6(\Omega)^3$ it follows that
$$
\frac{1}{2}\frac{d}{dt}\|\nabla \theta^n(t)\|_2^2   + \|\Delta \theta^n(t)\|_2^2 \leq \|u^n(t)\|_2\|\Delta \theta^n(t)\|_2 + c_2\|\nabla u^n(t)\|_2\|\|\nabla\theta^n(t)\|_{2}^{1/2}\|\Delta \theta^n(t)\|^{3/2}_2,
$$
with a constant $c_2 >0$. Now, the Young inequality implies
$$
\frac{d}{dt}\|\nabla \theta^n(t)\|_2^2   + \|\Delta \theta^n(t)\|_2^2 \leq 2\|u^n(t)\|^2_2 + c_3\|\nabla u^n(t)\|^4_2\|\|\nabla\theta^n(t)\|_{2}^{2},
$$
with a constant $c_3>0$. Assuming that $t\geq T_2$ and using Lemma \ref{lem:invariance} as well as the Poincar\'{e} inequality it follows that
$$
\frac{d}{dt}\|\nabla \theta^n(t)\|_2^2   + \pi^2\|\nabla\theta^n(t)\|_2^2 \leq \frac{8}{\pi\sqrt{6}}AD\mathrm{Ra}^2 + c_3\frac{4\pi}{\sqrt{6}}AD\mathrm{Ra}^2\|\nabla\theta^n(t)\|_{2}^{2}.
$$
We multiply by the integrating factor $e^{\pi^2 s}$, whence
$$
\frac{d}{dt}(\|\nabla \theta^n(s)\|_2^2e^{\pi^2 s}) \leq \frac{8}{\pi\sqrt{6}}AD\mathrm{Ra}^2e^{\pi^2 s} + c_3\frac{4\pi}{\sqrt{6}}AD\mathrm{Ra}^2\|\nabla\theta^n(s)\|_{2}^{2}e^{\pi^2 s}\quad \textrm{for a.e.}\quad s\geq T_2.
$$
Now assume that $t_2 > t_1 \geq \max\{ T_2, T_3\}$, where $T_3$ is given by  Lemma \ref{lemma:Temp_int}, and integrate the above inequality from $t_1$ to $t_2$. We obtain
$$
\|\nabla \theta^n(t_2)\|_2^2e^{\pi^2 t_2} \leq \|\nabla \theta^n(t_1)\|_2^2e^{\pi^2 t_1} + \frac{8}{\pi^3\sqrt{6}}AD\mathrm{Ra}^2e^{\pi^2 t_2} + c_3\frac{4\pi}{\sqrt{6}}AD\mathrm{Ra}^2\int_{t_1}^{t_2}\|\nabla\theta^n(s)\|_{2}^{2}e^{\pi^2 s}\, ds.
$$
Using Lemma \ref{lemma:Temp_int} we deduce
$$
\|\nabla \theta^n(t_2)\|_2^2 \leq \|\nabla \theta^n(t_1)\|_2^2e^{-\pi^2 (t_2-t_1)} + \frac{8}{\pi^3\sqrt{6}}AD\mathrm{Ra}^2 + c_3\frac{64\pi}{\sqrt{6}}A^2D\mathrm{Ra}^2 + c_3\frac{256}{\pi\sqrt{6}}A^2D^{3/2}\mathrm{Ra}^3.
$$
Now we choose $T_4 = \max\{  T_2, T_3 \} +1$. Let $t_2 \geq T_4$ and integrate the above inequality with respect to $t_1$ over the interval $[\max\{  T_2, T_3 \}, \max\{  T_2, T_3 \}+1]$. It follows that
$$
\|\nabla \theta^n(t_2)\|_2^2 \leq 16 A  + \frac{64}{\pi^2} A\sqrt{D}{\rm Ra} + \frac{8}{\pi^3\sqrt{6}}AD\mathrm{Ra}^2 + c_3\frac{64\pi}{\sqrt{6}}A^2D\mathrm{Ra}^2 + c_3\frac{256}{\pi\sqrt{6}}A^2D^{3/2}\mathrm{Ra}^3.
$$
The proof of \eqref{eq:grad-temp-est-H1} is complete. Now, \ref{eq:est-grad-temp-weak} follows from the fact that, for a subsequence of indexes $\theta^n(t) \to \theta(t)$ weakly in $E_1$ by Lemma \ref{lemma:convergence}, and, for another subsequence $\theta^n(t)\to a$ weakly in $W_1$ from boundedness in  \eqref{eq:grad-temp-est-H1}. Hence it must be $a = \theta(t)$, and the sequential weak lower semicontinuity of the norm ends the proof of \eqref{eq:est-grad-temp-weak}. 
\end{proof}
\end{lemma}
\section{Global attractors, their existence and invariance.} \label{sect:global}

\subsection{Existence of a global attractor.} From now on we make the standing assumption that the assumption (H) holds, i.e., that
$$
L \geq \frac{16 }{3\pi^2} K
\quad \textrm{and}\quad 
{\rm Pr} \geq 2c_1{\rm Ra}{D^{3/2}\sqrt{A}}.
$$
In what follows we will use Theorem \ref{thm:attr} with $X = H\times E_3\times E_1$ and $Y$ to be defined later. In fact if $Z = V \times W_3\times W_1$ then $Y$  will be a metric space given by a certain closed and bounded subset of $Z$ equipped with its topology. We need first to show that the multivalued map
\begin{align}
& S(t)(u_0,\gamma_0,\theta_0)  = \{ (u(t), \gamma(t), \theta(t)) \, :\ u,\gamma,\theta \quad \textrm{is a weak solution}\nonumber \\
	&\qquad \qquad \qquad \qquad  \textrm{given by Definition \ref{def-weak-sol-evol}  with initial data}\  (u_0,\gamma_0,\theta_0)\in X \},\label{eq:multivalued}
\end{align}
is a multivalued eventual semiflow, and that it satisfies the dissipativity and asymptotic compactness properties of Theorem \ref{thm:attr}. We start from the proof that it is a  multivalued eventual semiflow. 

\begin{lemma}\label{lem:multi}
	The family of multivalued mappings $\{ S(t) \}_{t\geq 0}$ defined by \eqref{eq:multivalued} is a multivalued eventual semiflow.
\end{lemma}
\begin{proof}
	Lemma \ref{lem:existence} implies that $S(t)(u_0,\gamma_0,\theta_0)$ is nonempty for every $(u_0,\gamma_0,\theta_0)\in X$. It is clear from the definition that $S(0)(u_0,\gamma_0,\theta_0) = \{ u_0,\gamma_0,\theta_0 \}$. We only have to prove the assertion (ii) of Definition \ref{def:eventual}. Consider the weak solution $(u(s),\gamma(s),\theta(s))$ given by Definition \ref{def-weak-sol-evol} and consider its suffix $$(u(s+t),\gamma(s+t),\theta(s+t))|_{s\in [0,\infty)}$$ where $t \geq t_1(B) = \max\{T_2,T_4\}$ with $T_2$ as in Lemma \ref{lem:invariance} and $T_4$ is as in Lemma \ref{le:grad-temp-est-H1}. The weak solution is a limit of the $(u,\gamma)$-Galerkin problems with strongly in $X$ converging initial data. It is clear that the regularity imposed in Definition \ref{def-weak-sol-evol} holds for the restrictions $(u(s-t),\gamma(s-t),\theta(s-t))|_{s\geq t}$ and they satisfy 
	the almost everywhere in time equations \eqref{eq:1f_w}--\eqref{eq:4f_w}. The restrictions of weak solution are the limits in the sense given in \eqref{conv_1}--\eqref{conv_3} of the restrictions of the approximative problems. Lemma  \ref{lemma:convergence} implies that, for a subsequence
	\begin{align*}
	& u^{n_k}(t) \to u(t)\quad \textrm{weakly in}\ H,\\
	& \gamma^{n_k}(t) \to \gamma(t)\quad \textrm{weakly in}\ E_3,\\
	& \theta^{n_k}(t) \to \theta(t)\quad \textrm{weakly in}\ E_1.	
	\end{align*} 
	But the estimates of Lemmas \ref{lem:invariance} and  \ref{le:grad-temp-est-H1} together with the compactness of the embeddings $V \subset H$, $W_3\subset E_3$, and $W_1\subset E_1$ imply that the above weak convergences are in fact strong. Hence, all requirements of Definition \ref{def:eventual} are satisfied by the family $\{ S(t) \}_{t\geq 0}$, and the proof is complete. 
\end{proof}

Now, the following dissipativity result is in a simple consequence of the previously derived a priori estimates. 

\begin{lemma}\label{lem:diss}
There exists a closed set $B_1 \in \mathcal{B}(Z)$ such that for every $B\in \mathcal{B}(X)$ there exists the time $t_1(B)$ such that 
$$
\bigcup_{t\geq t_1(B)}S(t)B \subset B_1
$$
\end{lemma}
\begin{proof}
	The assertion follows easily from Lemmas \ref{lem:enstrophy_2} and \ref{le:grad-temp-est-H1}. It is enough to take $t_1 = \max \{ T_2, T_4\}$. 
\end{proof}

In the next lemma we establish that the restriction of any weak solution to the interval $[t,\infty)$, where $t\geq t_1(B)$ is in fact strong.
\begin{lemma}\label{lem:strong}
	Let $(u,\gamma,\theta)$ be the weak solution with the initial data 
	$(u_0,\gamma_0,\theta_0) \in B \in \mathcal{B}(X)$. For every $t\geq t_1(B)$ the translated restrictions
	$$
	(u(\cdot+t),\gamma(\cdot+t),\theta(\cdot+t))|_{[0,\infty)}
	$$ 
	are the strong solutions.	
\end{lemma}
\begin{proof}
	Let $(u^n,\gamma^n,\theta^n)$ be the sequence of the $(u,\gamma)$-Galerkin solutions convergent to $(u,\gamma,\theta)$ in the sense given by Definition \ref{def-weak-sol-evol}, and let $t \geq t_1(B)$, where $T_2$ and $T_4$ are given by Lemmas \ref{lem:enstrophy_2} and \ref{le:grad-temp-est-H1}. 
	Denote 
	$$
(\overline{u},\overline \gamma,\overline \theta)  = (u(\cdot+t),\gamma(\cdot+t),\theta(\cdot+t))|_{[0,\infty)}
	$$
	$$
	(\overline{u}^n,\overline \gamma^n,\overline \theta^n)  = (u^n(\cdot+t),\gamma^n(\cdot+t),\theta^n(\cdot+t))|_{[0,\infty)}
	$$
	By Lemmas \ref{lem:invariance}, \ref{lem:enstrophy_2} and \ref{le:grad-temp-est-H1} there exists a constant $C$ independent of $n$ such that 
\begin{equation}\label{eq:grad_bound}
\|\nabla \overline u^n(s)\|_2^2+ \|\nabla \overline \gamma^n(s)\|_2^2 + \|\nabla \overline \theta^n(s)\|_2^2 \leq C\quad \textrm{for}\quad s\geq 0,
\end{equation}
and
\begin{equation}
\|\nabla \overline u(s)\|_2^2+ \|\nabla \overline \gamma(s)\|_2^2 + \|\nabla \overline \theta(s)\|_2^2 \leq C\quad \textrm{for}\quad s\geq 0.
\end{equation}
Functions $(\overline{u}^n,\overline{\gamma}^n,\overline{\theta}^n)$ satisfy the equations
\begin{align}
& \epsilon\left(((\overline{u}^n)_t(s),v) + ((\overline{u}^n(s)\cdot \nabla) \overline{u}^n(s),v)\right) + (1+K)(\nabla \overline{u}^n(s),\nabla v) = 2K ({\rm rot}\, \overline{\gamma}^n(s),v) + {\rm Ra} (\overline{\theta}^n(s),v_3),\label{eq:galer_1}\\
& \epsilon M\left(((\overline{\gamma}^n)_t(s),\xi) + ((\overline{u}^n(s)\cdot \nabla)\overline{\gamma}^n(s),\xi)\right) + L(\nabla \overline{\gamma}^n(s),\nabla \xi) + G({\rm div}\,\overline{\gamma}^n(s), {\rm div}\,\xi) +
4 K (\overline{\gamma}^n(s),\xi)\nonumber\\
&\qquad = 2K ({\rm rot}\, \overline{u}^n(s),\xi),\label{eq:galer_2}\\\
& \langle (\overline{\theta}^n)_t(s), \eta \rangle + (\overline{u}^n(s)\cdot \nabla \overline{\theta}^n(s),\eta) + (\nabla  \overline{\theta}^n(s),\nabla \eta)  = ((\overline{u}^n)_3(s),\eta), \label{eq:galer_3}\
\end{align}
for any $(v,\xi,\eta)\in H^n\times E^n_3\times E_1$. Since, in the following estimates we do not need exactly to keep tract of the dependence on the constants of the problem, we will denote by $C$ the generic constant independent of $n$ and $k$. Testing \eqref{eq:galer_1} by $v = -P\Delta \overline{u}^n(t) $ and proceeding the same as in the proof of Lemma \ref{lem:enstrophy} we obtain the estimate
\begin{align*}
& \frac{\epsilon}{2}\frac{d}{dt}\|\nabla \overline{u}^n(t)\|^2_2+  (1+K)\|P\Delta \overline{u}^n(t)\|_2^2\nonumber \\
& \qquad \leq  2K\|\nabla \overline{\gamma}^n(t)\|_2\|P\Delta \overline{u}^n(t)\|_2 + {\rm Ra} \|\overline{\theta}^n(t)\|_2\|P\Delta \overline{u}^n(t)\|_2 + C\epsilon\|\nabla \overline{u}^n(t)\|^{3/2}_2\|P\Delta \overline{u}^n(t)\|^{3/2}_2.
\end{align*}
Estimates \eqref{eq:grad_bound} as well as the Cauchy and Young inequalities imply that 
$$
\frac{\epsilon}{2}\frac{d}{dt}\|\nabla \overline{u}_n^k(t)\|^2_2+  \|P\Delta \overline{u}_n^k(t)\|_2^2 \leq C.
$$
This means that
\begin{equation}\label{eq:e_aux_1}
\int_0^T \|P\Delta \overline{u}_n^k(t)\|_2^2\, dt \leq C(1+T). 
\end{equation}
Sequential weak lower semicontinuity of the norm in the space $L^2(0,T;D(-P\Delta))$, where $D(-P\Delta)$ is the domain of the Stokes operator $-P\Delta$ equipped with the $H^2$ norm, implies that $\overline{u} \in L^2(0,T;D(-P\Delta))$ also 
\begin{equation}\label{eq:e_aux_2}
\int_0^T \|P\Delta \overline{u}(t)\|_2^2\, dt \leq C(1+T). 
\end{equation}
In a standard way, from \eqref{eq:galer_1}, testing it by $v(t)$ with arbitrary $v\in L^2(0,T;H)$, estimating the nonlinear term using the Agmon inequality, and using the previous estimates we obtain the bound
\begin{equation}\label{eq:e_aux_3}
\int_0^T \|(\overline{u}^n)_t(t)\|_2^2\, dt \leq C(1+T).  
\end{equation}
Again, the sequential weak lower semicontinuity of the norm implies that 
$\overline{u}_t \in L^2(0,T;H)$, and that
\begin{equation}\label{eq:e_aux_4}
\int_0^T \|\overline{u}_t(t)\|_2^2\, dt \leq C(1+T).  
\end{equation}
In a similar way, analogous estimates on microrotation $\gamma$ and temperature $\theta$ are obtained from \eqref{eq:galer_2} and \eqref{eq:galer_3}. These estimates give us enough compactness to pass to the limit with $n\to \infty$  in \eqref{eq:galer_1}--\eqref{eq:galer_3}, whence $(\overline{u},\overline \gamma,\overline \theta)$ is a strong solution. We only show how to pass to the limit in the nonlinearm term 
$((\overline u^n\cdot\nabla \overline u^n),v)$. 
We know that, for a function $u\in L^2(0,T;D(-P\Delta))\cap H^1(0,T;H)$ there holds
\begin{align}
	& \overline{u}^n \to u \quad \textrm{weakly in}\quad L^2(0,T;D(-P\Delta)) \quad \textrm{and strongly in}\quad L^2(0,T;V\cap W^{1,3}(\Omega)^3),\label{conv_str1}\\
& \overline{u}^n_t \to u_t \quad \textrm{weakly in}\quad L^2(0,T;H),
\end{align}
where the second convergence in \eqref{conv_str1} follows from the Aubin--Lions lemma and the compact embedding $D(-P\Delta) \subset W^{1,3}(\Omega)^3$.  
Taking $v\in L^2(0,T;H)$ we obtain
\begin{align*}
& \int_0^T((\overline u^n(s)\cdot\nabla) \overline u^n(s)-( u(s)\cdot\nabla )  u(s),v(s))\, ds\\
&\quad  = \int_0^T( (\overline u^n(s)\cdot\nabla) (\overline u^n(s)-u(s)),v(s))\, ds + \int_0^T( ((\overline u^n(s)-u(s))\cdot\nabla)  u(s),v(s))\, ds = I + II.
\end{align*}
We estimate both terms separately
\begin{align*}
& |I| \leq \int_0^T \|\overline u^n(s)\|_{6} \|\nabla (\overline u^n(s)-u(s))\|_3 \|v(s)\|_2\, ds \\
& \quad \leq C\int_0^T \|\nabla \overline u^n(s)\|_{2} \|\nabla (\overline u^n(s)-u(s))\|_3 \|v(s)\|_2\, ds\\
& \quad  \leq C\|\overline u^n-u\|_{L^2(0,T;V\cap W^{1,3}(\Omega)^3)}\|v\|_{L^2(0,T;H)} \to 0,
\end{align*}
\begin{align*}
& |II| \leq \int_0^T \|\overline u^n(s)-u(s)\|_{\infty}\|\nabla  u(s)\|_2 \|v(s)\|_2\, ds \\
& \quad \leq C\int_0^T  \|\nabla (\overline u^n(s)-u(s))\|_2^{1/2}\|\overline u^n(s)-u(s)\|_{D(-P\Delta)}^{
1/2} \|v(s)\|_2\, ds\\
& \quad  \leq C\|v\|_{L^2(0,T;H)}\sqrt{\int_0^T  \|\nabla (\overline u^n(s)-u(s))\|_2\|\overline u^n(s)-u(s)\|_{D(-P\Delta)} \, ds}\\
& \quad \leq C\|v\|_{L^2(0,T;H)} \|\overline u^n-u\|^{1/2}_{L^2(0,T;V)}\left(\|\overline u^n\|_{L^2(0,T;D(-P\Delta))}+\|u\|_{L^2(0,T;D(-P\Delta))}\right)^{1/2} \to 0.
\end{align*}
 It is clear that we can pass to the limit in the nonlinear term in \eqref{eq:galer_1}. We skip details of passing to the limit in the remaining terms. In the linear ones the possibility of pass to the limit just follows from the weak convergence, and passing to the limit in nonlinear terms in \eqref{eq:galer_2} and \eqref{eq:galer_3} is done in a similar way as in the one in \eqref{eq:galer_1}. The the proof is complete. 
\end{proof}

We are in position to define the metric space $Y$. It is given by
$$
Y = \overline{\bigcup_{t\geq t_1(B_1)}S(t)B_1}^Z,
$$
which, equipped with the norm topology of $Z$ is a complete metric space, as a closed subset of $Z$. Lemma \ref{lem:diss} implies that $Y \subset B_1$, whence $Y$ is bounded. 

In the next lemma we show that $Y$ is absorbing, and hence it can be used as $B_0$ in Theorem \ref{thm:attr}.

\begin{lemma}\label{lemma_absorbing_attr}
	For every $B\subset \mathcal{B}(X)$ there exists $t_0=t_0(B)$ such that 
	$$
	\bigcup_{t\geq t_0}S(t)B \subset Y.
	$$
\end{lemma}
\begin{proof}
	Define $t_0 = t_1(B) + t_1(B_1)$ and take $t\geq t_0$.
	Clearly $t \geq  t_1(B)$ and hence, by Lemma \ref{lem:multi}
	$$
	S(t)B = S(t-t_1(B) + t_1(B))B \subset S(t-t_1(B))S(t_1(B))B. 
	$$
	We now use Lemma \ref{lem:diss} to deduce that
	$$
	S(t)B \subset S(t-t_1(B))B_1. 
	$$
	But $t-t_1(B) \geq t_1(B_1)$, and hence
	$$
	S(t)B \subset Y,
	$$
	and the proof is complete.
\end{proof}

We pass to the proof of the asymptotic compactness. The technique to prove it is based on the energy equation method, see for instance \cite{Ball}. Note that in Theorem 
\ref{thm:attr} it is only sufficient to obtain the asymptotic compactness for the initial data in the absorbing set $B_0$. We will in fact obtain the asymptotic compactness for the initial data in any set which is bounded in $X$.

\begin{lemma}\label{lem:comp}
	Assume that $B\in \mathcal{B}(X)$ and $(w_n,\xi_n,\eta_n) \in S(t_n)B$ with a sequence $t_n\to \infty$. Then the sequence $(w_n,\xi_n,\eta_n)$ is relatively compact in $Y$, i.e. $w_n$ is relatively compact in $V$, $\xi_n$ is relatively compact in $W_3$, and $\eta_n$ is relatively compact in $W_1$.
\end{lemma}
\begin{proof}
	There exists a sequence $(u_{0n},\gamma_{0n},\theta_{0n}) \in B$ and a sequence $(u_n,\gamma_n,\theta_n)$ of weak solutions with the initial data $(u_{0n},\gamma_{0n},\theta_{0n})$ such that $(u_n(t_n),\gamma_{n}(t_n),\theta_{n}(t_n)) = (w_n,\xi_n,\eta_n)$. We will consider the restrictions $$(\overline{u}_n,\overline{\gamma}_n,\overline{\theta}_n) = (u_n(t+t_n-1),\gamma_n(t+t_n-1),\theta_n(t+t_n-1))|_{t\in [0,2]}.$$
	These functions are defined on the time interval $[0,2]$, and
	$$(\overline{u}_n(1),\overline{\gamma}_n(1),\overline{\theta}_n(1)) = (w_n,\xi_n,\eta_n).$$ 
	
	By Lemma \ref{lem:strong}, if only $n$ is large enough, the functions $(\overline{u}_n,\overline{\gamma}_n,\overline{\theta}_n)$ are strong solutions on the interval $[0,2]$. Moreover, by Lemma \ref{lem:diss} there exists a constant $C>0$ such that
	\begin{equation}\label{eq:bound_grad}
	\|\nabla \overline{u}_n(t)\|_2^2+ \|\nabla \overline{\gamma}_n(t)\|_2^2 + \|\nabla \overline{\theta}_n(t)\|_2^2 \leq C\quad \textrm{for}\quad t\in [0,2].
	\end{equation}
	Estimate \eqref{eq:bound_grad} implies that, for a nonrenumbered subsequence of indexes,
	\begin{align}
	&\overline{u}_n(1) \to a \quad \textrm{weakly in}\quad V,\label{conv_weak_1}\\
	&\overline{\gamma}_n(1) \to b\quad \textrm{weakly in}\quad W_3\label{conv_weak_2},\\
	&\overline{\theta}_n(1) \to c\quad \textrm{weakly in}\quad W_1,\label{conv_weak_3}	
	\end{align} 
	for some $(a,b,c)\in Y$. We have to show that these convergences are in fact strong. Functions $(\overline{u}^n,\overline{\gamma}^n,\overline{\theta}^n)$ satisfy the equations \eqref{eq:galer_1}--\eqref{eq:galer_3} for almost every $s\in [0,2]$. Hence, proceeding exactly as in the proof of Lemma \ref{lem:strong} we get the estimates 
	$$
	\int_0^2 \|P\Delta \overline{u}^n(t)\|_2^2 + \|\Delta \overline{\gamma}^n(t)\|_2^2 + \|\Delta \overline{\theta}^n(t)\|_2^2\, dt \leq C, 
	$$
	and 
	$$
	\int_0^2 \|\overline{u}^n_t(t)\|_2^2 + \|\overline{\gamma}^n_t(t)\|_2^2 + \|\overline{\theta}^n_t(t)\|_2^2\, dt \leq C.  
	$$
	 Above estimates, together with the Aubin--Lions lemma, imply that, for a subsequence of $n$, we have the following convergences
	 \begin{align*}
	& \overline{u}^n \to u \quad \textrm{weakly in}\quad L^2(0,2;D(-P\Delta)) \quad \textrm{and strongly in}\quad L^2(0,2;V),\\
	& \overline{u}^n_t \to u_t \quad \textrm{weakly in}\quad L^2(0,2;H),\\
	& \overline{\gamma}^n \to \gamma \quad \textrm{weakly in}\quad L^2(0,2;D_3(-\Delta)) \quad \textrm{and strongly in}\quad L^2(0,2;W_3),\\
	& \overline{\gamma}^n_t \to \gamma_t \quad \textrm{weakly in}\quad L^2(0,2;H),\\
	& \overline{\theta}^n \to \theta\quad \textrm{weakly in}\quad L^2(0,2;D_1(-\Delta)) \quad \textrm{and strongly in}\quad L^2(0,2;W_1),\\
	& \overline{\theta}^n_t \to \theta_t \quad \textrm{weakly in}\quad L^2(0,2;E_1),
	\end{align*}
	where
	 \begin{align*}
	 & u\in L^2(0,2;D(-P\Delta)) \cap C([0,2];V)\quad \textrm{with}\quad u_t\in L^2(0,2;H),\\
 	 & \gamma\in L^2(0,2;D_3(-\Delta)) \cap C([0,2];W_3)\quad \textrm{with}\quad \gamma_t\in L^2(0,2;E_3),\\
	  & \theta\in L^2(0,2;D_1(-\Delta)) \cap C([0,2];W_1)\quad \textrm{with}\quad u_t\in L^2(0,2;E_1).
	 \end{align*}  
	 In particular 
	 \begin{align*}
	 & \overline{u}_n(t) \to u(t)\quad \textrm{strongly in}\quad V\quad \textrm{for a.e.}\quad t\in (0,2),\\
	 & \overline{\gamma}_n(t) \to \gamma(t)\quad \textrm{strongly in}\quad W_3\quad \textrm{for a.e.}\quad t\in (0,2),\\
	 & \overline{\theta}_n(t) \to \theta(t)\quad \textrm{strongly in}\quad W_1\quad \textrm{for a.e.}\quad t\in (0,2),
	 \end{align*}
	 and 
	 	 \begin{align*}
	 & \overline{u}_n(t) \to u(t)\quad \textrm{weakly in}\quad V\quad \textrm{for every}\quad t\in [0,2],\\
	 & \overline{\gamma}_n(t) \to \gamma(t)\quad \textrm{weakly in}\quad W_3\quad \textrm{for every}\quad t\in [0,2],\\
	 & \overline{\theta}_n(t) \to \theta(t)\quad \textrm{weakly in}\quad W_1\quad \textrm{for every}\quad t\in [0,2].
	 \end{align*}
	 Coming back to \eqref{conv_weak_1}--\eqref{conv_weak_3} we deduce that $a=u(1)$, $b=\gamma(1)$, $c=\theta(1)$. We will show that $\|\nabla \overline u_n(1)\|_2 \to \|\nabla u(1)\|_2$. This will mean that the convergence in \eqref{conv_weak_1} is in fact strong. Since the proofs that convergences in \eqref{conv_weak_2} and \eqref{conv_weak_3} are strong are analogous, and the technique is well known, we will only provide the proof for \eqref{conv_weak_1}. Testing \eqref{eq:galer_1} with $v=-P\overline u_n(s)$ and integrating from $0$ to $t\in (0,2)$ we get
\begin{align*}
& \frac{\epsilon}{2}\|\nabla \overline u_n(t)\|_2^2 \\
& \qquad - \int_0^t{\rm Ra} (\overline \theta_n(s),(P\Delta  \overline u_n)_3(s)) + \epsilon((\overline u_n(s)\cdot \nabla) \overline u_n(s),P\Delta \overline u_n(s)) - 2K ({\rm rot}\, \overline\gamma_n(s),P\Delta \overline u_n(s))\, ds\\
& \qquad = \frac{\epsilon}{2}\|\nabla \overline u_n(0)\|_2^2 - (1+K)\int_0^t \|P \Delta \overline u_n(s)\|_2^2\, ds.
\end{align*}
Denote 
\begin{align*}
& W_n(t) = \frac{\epsilon}{2}\|\nabla \overline u_n(t)\|_2^2 \\
& \qquad - \int_0^t{\rm Ra} (\overline \theta_n(s),P\Delta  (\overline u_n)_3(s)) + \epsilon((\overline u_n(s)\cdot \nabla) \overline u_n(s),P\Delta \overline u_n(s)) - 2K ({\rm rot}\, \overline\gamma_n(s),P\Delta \overline u_n(s))\, ds,
\end{align*}
and 
$$
W(t) = \frac{\epsilon}{2}\|\nabla  u(t)\|_2^2 - \int_0^t{\rm Ra} ( \theta(s),(P\Delta  u)_3(s)) + \epsilon((u(s)\cdot \nabla)  u(s),P\Delta  u(s)) - 2K ({\rm rot}\, \gamma(s),P\Delta  u(s))\, ds.
$$
It is not hard to verify that $W_n(t)$ are nonincreasing functions of time, $W_n(t) \to W(t)$ for almost every $t\in (0,2)$, and function $W(t)$ is continuous. This implies that $W_n(t) \to W(t)$ for every $t\in (0,1)$, whence in particular $W_n(1) \to W(1)$, and $\|\nabla \overline u_n(1)\|_2 \to \|\nabla u(1)\|_2$, which completes the proof. 
\end{proof}

Summarizing, Theorem \ref{thm:attr}, as well as Lemmas \ref{lem:multi}, \ref{lemma_absorbing_attr}, and \ref{lem:comp} imply the following result.

\begin{theorem}\label{thm:existence}
The multivalued eventual semiflow $\{ S(t) \}_{t\geq 0}$ defined by \eqref{eq:multivalued} has a $(X,Y)$-global attractor $\mathcal{A}$.  
\end{theorem}

The above theorem shows the existence of a global attractor in the sense of Definition \ref{def:xyattr}. This attractor is the smallest compact attracting set, but not necessarily invariant. The obtained attractor $\mathcal{A}$ is, however, an invariant set, and the semiflow restricted to this attractor is in fact single-valued and governed by the strong solutions. These results will be proved in the next subsection.

\subsection{Invariance of the global attractor $\mathcal{A}$.} We start from the observation which follows from Lemma \ref{lem:comp} and the weak-strong uniqueness property obtained in Lemma \ref{lemma:weak-strong}.

\begin{lemma}\label{lemma:strong}
	Let $(u,\gamma,\theta)$ be a weak solution given by Definition \ref{def-weak-sol-evol}. There exists $t_0$, which can be chosen uniformly with respect to bounded sets in $X$ of initial data, such that for every $t\geq t_0$, this solution restricted to $[t,\infty)$ is in fact strong. Moreover, for any $t\geq T$ and any $s\geq 0$ the set $S(s)(u(t),\gamma(t),\theta(t))$ is a singleton contained in $Y$. 
\end{lemma}

Note that in the above theorem the time $t_0(B)$ is exactly the same as in Lemmma \ref{lemma_absorbing_attr}. In the next Lemma we establish the result on the continuous dependence of the strong solutions on the initial data. 
\begin{lemma}\label{lem:continuity}
	Let $(u^1,\gamma^1,\theta^1)$ and $(u^2,\gamma^2,\theta^2)$ be two strong solutions such that 
	\begin{align}
	&\| (u^1(t), \gamma^1(t), \theta^1(t)) \|_Z \leq C\quad \textrm{for}\quad t\in [0,T],\label{eq:est_C_1}\\
	&\| (u^2(t), \gamma^2(t), \theta^2(t)) \|_Z \leq C\quad \textrm{for}\quad t\in [0,T].\label{eq:est_C_2}
	\end{align}
	Then there exists a constant $L$ dependent only of $C$ and $T$ such that
	$$
	\| (u^1(t), \gamma^1(t), \theta^1(t)) - (u^2(t), \gamma^2(t), \theta^2(t)) \|_Z \leq L \| (u^1(0), \gamma^1(0), \theta^1(0)) - (u^2(0), \gamma^2(0), \theta^2(0)) \|_Z,
	$$
	for every $t\in [0,T]$. 
\end{lemma} 
\begin{proof} In the proof by $E$ we will denote the generic constant dependent only on the constants present in the problem definition and $C$ from \eqref{eq:est_C_1} and \eqref{eq:est_C_2}. 
	Denote $w = u^1 - u^2$, $\xi = \gamma^1 - \gamma^2$, and $\eta = \theta^1-\theta^2$.  Testing the difference of  \eqref{eq:strong_1} written for $u^1$ and $u^2$ by $-P\Delta w$ we obtain
	\begin{align*}
	& \frac{\epsilon}{2}\frac{d}{dt}\|\nabla w(t)\|_2^2  + (1+K)\|P \Delta w(t)\|_2^2 \\
	& \qquad = -2K ({\rm rot}\, \xi(t),P\Delta w(t)) - {\rm Ra} (\eta(t),(P\Delta w)_3(t)) \\
	&\qquad \qquad + \epsilon((w(t)\cdot \nabla) u^1(t),P\Delta w(t))+\epsilon((u^2(t)\cdot \nabla) w(t),P\Delta w(t)).
	\end{align*}
	Using the Schwartz and H\"older inequalities, we deduce
	\begin{align*}
	& \frac{\epsilon}{2}\frac{d}{dt}\|\nabla w(t)\|_2^2  + (1+K)\|P \Delta w(t)\|_2^2 \\
	& \qquad \leq E\|P\Delta w(t)\|_2\left( \|\nabla \xi(t)\|_2 +  \|\eta(t)\|_2+ \|w(t)\|_\infty \|\nabla u^1(t)\|_2 +\|u^2(t)\|_6 \|\nabla w(t)\|_3\right).
	\end{align*}
	Using the Agmon inequality $\|w\|_\infty \leq E \|\nabla w\|_2^{1/2}\|P\Delta w\|_2^{1/2}$ and the interpolation inequality $\|\nabla w\|_3 \leq E \|\nabla w\|_2^{1/2}\|P\Delta w\|_2^{1/2}$ as well as the continuity of the embedding $V\subset L^6(\Omega)^3$ and  assumptions \eqref{eq:est_C_1}--\eqref{eq:est_C_2}, we obtain
	we deduce
	\begin{align*}
	& \frac{d}{dt}\|\nabla w(t)\|_2^2  + (1+K)\|P \Delta w(t)\|_2^2 \\
	& \qquad \leq E\|P\Delta w(t)\|_2\left( \|\nabla \xi(t)\|_2 +  \|\eta(t)\|_2\right)+ E\|\nabla w(t)\|_2^{1/2}\|P\Delta w(t)\|_2^{3/2} + E\|\nabla w(t)\|_2^{1/2}\|P\Delta w(t)\|_2^{3/2}.
	\end{align*}
	We can use the Young inequality with $\epsilon$, as well as the Poincar\'{e} inequality, whence we obtain
	$$
	\frac{d}{dt}\|\nabla w(t)\|_2^2   \leq E\left( \|\nabla \xi(t)\|^2_2 +  \|\nabla \eta(t)\|^2_2+ \|\nabla w(t)\|^2_2\right).
	$$
	Proceeding in a similar way with \eqref{eq:strong_2}, testing the difference of this equation written for $\gamma^1$ and $\gamma^2$ with $-\Delta \xi(t)$  we obtain,
	$$
	\frac{d}{dt}\|\nabla \xi(t)\|_2^2   \leq E\left( \|\nabla \xi(t)\|^2_2 +   \|\nabla w(t)\|^2_2\right).
	$$
	Finally, proceeding analogously with \eqref{eq:strong_3} we arrive at
	$$
	\frac{d}{dt}\|\nabla \eta(t)\|_2^2   \leq E\left( \|\nabla u(t)\|^2_2 +   \|\nabla \eta(t)\|^2_2\right).
	$$
	Adding the three obtained inequalities to each other we get
	$$
	\frac{d}{dt}\|(w(t),\xi(t),\eta(t))\|_Z^2 \leq E \|(w(t),\xi(t),\eta(t))\|_Z^2,
	$$
	where the constant $E$ depends only on the problem data and $C$. So, by the Gronwall lemma, the assertion holds with $L=e^{ET}$.  
\end{proof}
In the next result we establish that the multivalued eventual semiflow given by weak solutions is in fact a single valued semiflow in $Y$ when we restrict it to $Y$.

\begin{lemma}\label{lem:strong_on_attractor}
	Let $(u_0,\gamma_0,\theta_0) \in Y$. There exists a strong solution $(u(t),\gamma(t),\theta(t))$ with the initial data $(u_0,\gamma_0,\theta_0)$ which is also a unique weak solution, and for every $t\geq 0$ there holds $(u(t),\gamma(t),\theta(t))\in Y$. In consequence, $S(t)$ restricted to $Y$ is a single valued semiflow and has values in $Y$. 
\end{lemma}
\begin{proof}
	Let $(u_0,\gamma_0,\theta_0) \in Y$. Then either
	$$
	(u_0,\gamma_0,\theta_0) \in \bigcup_{t\geq t_1(B_1)} S(t)B_1\quad \textrm{or}\quad 	(u_0,\gamma_0,\theta_0) \in \overline{\bigcup_{t\geq t_1(B_1)} S(t)B_1}^Z \setminus \bigcup_{t\geq t_1(B_1)} S(t)B_1
	$$
	In the first case, by Lemma \ref{lem:strong} there exists the strong solution starting from $(u_0,\gamma_0,\theta_0)$, which, by Lemma \ref{lemma:weak-strong} must be unique in the class of the weak solutions. Hence $S(t)(u_0,\gamma_0,\theta_0)$ is a singleton and moreover  $S(t)(u_0,\gamma_0,\theta_0) \in Y$. It remains to verify the second possibility. 
	In that case there exists the sequence $\{t_n\} \subset [t_1(B_1),\infty)$ and the sequence of solutions $\{ (u^n,\gamma^n,\theta^n) \}_{n=1}^\infty$ with the initial data in $B_1$ such that $(u^n(t_n),\gamma^n(t_n),\theta^n(t_n)) \to (u_0,\gamma_0,\theta_0)$ in $Z$. 
	Consider the translations 
	$$
	(\overline{u}^n(s),\overline{\gamma}^n(s),\overline{\theta}^n(s)) = (u^n(t_n+s),\gamma^n(t_n+s),\theta^n(t_n+s))\quad \textrm{for}\quad t\in [0,\infty).
	$$
	Lemma \ref{lemma:strong} implies that these are the strong solutions.
	Moreover, arguing as in the proof of Lemma \ref{lem:strong} we have the bounds
$$
	\|\nabla \overline u^n(s)\|_2^2+ \|\nabla \overline \gamma^n(s)\|_2^2 + \|\nabla \overline \theta^n(s)\|_2^2 \leq C\quad \textrm{for every}\quad s\geq 0,
$$
$$
\int_0^T \|P\Delta \overline{u}^n(s)\|_2^2 + \|\Delta \overline{\gamma}^n(s)\|_2^2 + \|\Delta \overline{\theta}^n(s)\|_2^2\, ds \leq C(T) \quad \textrm{for every}\quad T\geq 0, 
$$
and 
$$
\int_0^T \|\overline{u}^n_t(s)\|_2^2 + \|\overline{\gamma}^n_t(s)\|_2^2 + \|\overline{\theta}^n_t(s)\|_2^2\, ds \leq C(T) \quad \textrm{for every}\quad T\geq 0.  
$$
Analogously as in the proofs of Lemmas \ref{lem:strong} and \ref{lem:comp}, and using a diagonal argument, we can extract a subsequence such that 
	 \begin{align*}
& \overline{u}^n \to u \quad \textrm{weakly in}\quad L^2_{loc}([0,\infty);D(-P\Delta)) \quad \textrm{and strongly in}\quad L^2_{loc}([0,\infty);V\cap W^{1,3}(\Omega)^3),\\
& \overline{u}^n_t \to u_t \quad \textrm{weakly in}\quad L^2_{loc}([0,\infty);H),\\
& \overline{\gamma}^n \to \gamma \quad \textrm{weakly in}\quad L^2_{loc}([0,\infty);D_3(-\Delta)) \quad \textrm{and strongly in}\quad L^2_{loc}([0,\infty);W_3\cap W^{1,3}(\Omega)^3),\\
& \overline{\gamma}^n_t \to \gamma_t \quad \textrm{weakly in}\quad L^2_{loc}([0,\infty);H),\\
& \overline{\theta}^n \to \theta\quad \textrm{weakly in}\quad L^2_{loc}([0,\infty);D_1(-\Delta)) \quad \textrm{and strongly in}\quad L^2_{loc}([0,\infty);W_1\cap W^{1,3}(\Omega)),\\
& \overline{\theta}^n_t \to \theta_t \quad \textrm{weakly in}\quad L^2_{loc}([0,\infty);E_1),
\end{align*}
where $(u,\gamma,\theta)$ is a strong solution with the initial data $(u_0,\gamma_0,\theta_0)$. Weak-strong uniqueness obtained in Lemma \ref{lemma:weak-strong} implies that $S(t)(u_0,\gamma_0,\theta_0)$ is a singleton. Sequential weak lower semicontinuity of the norm implies that 
$$
\|\nabla  u(s)\|_2^2+ \|\nabla  \gamma(s)\|_2^2 + \|\nabla  \theta(s)\|_2^2 \leq C\quad \textrm{for every}\quad s\geq 0.
$$
Hence, by Lemma \ref{lem:continuity} we deduce that
\begin{align*}
&(u^n(t_n+s),\gamma^n(t_n+s),\theta^n(t_n+s)) = (\overline u^n(s), \overline \gamma^n(s), \overline \theta^n(s)) \to ( u(s),  \gamma(s),  \theta(s))\\
&\quad\textrm{in}\quad Z \quad \textrm{as}\quad n\to \infty\quad \textrm{for every}\quad s\geq 0. 
\end{align*}
This means that 
$$S(t)(u_0,\gamma_0,\theta_0) = ( u(s),  \gamma(s),  \theta(s)) \in Y,$$
and the proof is complete.
\end{proof}

We have all ingredients ready for the result on the attractor invariance. 
\begin{lemma}\label{lem:invariance_2}
Let $\mathcal{A}$ be the $(X,Y)$-global attractor obtained in Theorem \ref{thm:existence}. For every $t\geq 0$ there holds $S(t)\mathcal{A} = \mathcal{A}$, and $S(t)$ is a single valued semigroup on $\mathcal{A}$. 
\end{lemma}
\begin{proof}
	As $\mathcal{A}$ is a compact set contained in $Y$, and by Lemma \ref{lem:strong_on_attractor}, $S(t)$ is single-valued on $Y$, it also must be that $S(t)$ is single valued on $\mathcal{A}$. The fact that on $Y$ the multivalued maps $S(t)$ are actually governed by the strong solutions and have valued in $Y$ implies that $S(t)|_Y$ is a semigroup.  Since $Y$ is a bounded set in $Z$, Lemma \ref{lem:continuity} implies that $S(t)|_Y$ are $(Y,Y)$ continuous maps. Hence,  Theorem \ref{thm:invariance_abstract} implies that $\mathcal{A}$ is an invariant set. The proof is complete.   
\end{proof}

\begin{remark}
	We have only proved that the attractor $\mathcal{A}$ is a compact set in the topology of $Z = V\times W_3\times W_1$. We note that it is possible to continue the bootstrapping argument, which would lead us to further regularity of the attractor $\mathcal{A}$ in higher order Sobolev spaces. It is also possible, proceeding in a now well established way to get its finite dimensionality. This would give an example of a problem without known solution uniqueness, which has the attractor of finite dimensionality. The solutions are unique on the attractor, but, in contrast to examples of \cite{KLS2}, the semigroup does not instantaneously enter the regime where the solution has to stay unique. Rather than that, such regime is entered after some time given uniformly with respect to bounded sets of initial data. 
\end{remark}

\section{Upper semicontinuous convergence of attractors.} \label{sect:upper}
The main aim of this section is the proof of the following convergence
\begin{equation}\label{eq:hausdorff}
\lim_{K\to 0_+}\mathrm{dist}_X(\mathcal{A}^K,\mathcal{A}^0) = 0.
\end{equation}
In this section we will assume that the parameters ${\rm Pr}$, ${\rm Ra}$, $M, L, G, A$ are fixed and satisfy (H), and the parameter $K$ varies
in the interval
$$
K \in \left[ 0, L\frac{3\pi^2}{16 }\right],
$$
such that the assumption (H) always holds. By $C$ we will denote a generic constant independent of $K$, but possibly dependent on parameters ${\rm Pr}$, ${\rm Ra}$, $M, L, G, A$. We also denote $L\frac{3\pi^2}{16 } = K_{max}$. For every $K \in [0,K_{max}]$ we denote the corresponding global attractor, which exists by Theorem \ref{thm:existence} and is invariant by Lemma \ref{lem:invariance_2} by $\mathcal{A}^K$. The corresponding semigroup will be denoted by $\{ S(t)^K \}_{t\geq 0}$. Note that while the Banach spaces $X = H\times E_3\times E_1$ and $Z = V\times W_3\times W_1$ are independent of $K$, the metric space $Y$ is dependent of $K$. We will denote it by $Y^K$.

\subsection{Hausdorff and Kuratowski upper-semicontinuous convergence.} We start from the definitions of Hausdorff and Kuratowski upper semicontinuous convergence.

\begin{definition}
	Let $(X,\varrho)$ be a metric space and let $\{ A_K \}_{ K \in [0,K_{max}]}$ be sets in $X$. We say that the family $\{ A_K \}_{ K \in (0,K_{max}]}$ converges to $A_0$ upper-semicontinuously in Hausdorff sense if 
	$$
	\lim_{K\to 0^+}\mathrm{dist}_X(A_K,A_0) = 0.
	$$
\end{definition}
  
\begin{definition}
	Let $(X,\varrho)$ be a metric space and let $\{ A_K \}_{ K \in [0,K_{max}]}$ be sets in $X$. We say that the family $\{ A_K \}_{ K \in (0,K_{max}]}$ converges to $A_0$ upper-semicontinuously in Kuratowski sense if 
	$$
	X-\limsup_{K\to 0^+} A_K \subset A_0,
	$$
	where $X-\limsup_{K\to 0^+} A_K$ is the Kuratowski upper limit defined by
	$$
		X-\limsup_{K\to 0^+} A_K = \{ x\in X\,:\  \lim_{n\to \infty}\rho(x_n,x) = 0, x_n \in A_{K_n}, K_n\to 0\ \textrm{as}\ n\to \infty \}.
	$$
\end{definition}  
We relate the two convergences by the following well known result, cf., \cite[Proposition 4.7.16]{DMP}.
 \begin{proposition}\label{prop:upper}
 	Assume that the sets $\{ A_K \}_{ K \in [0,K_{max}]}$ are nonempty and compact and the set $\bigcup_{ K \in (0,K_{max}]} A_K$ is relatively compact.  
 	If the family $\{ A_K \}_{ K \in (0,K_{max}]}$ converges to $A_0$ upper-semicontinuously in Kuratowski sense then  	$\{ A_K \}_{ K \in (0,K_{max}]}$ converges to $A_0$ upper-semicontinuously in Hausdorff sense.
 \end{proposition} 
\begin{proof}
Since all $A_K$ are compact sets then for every $K$ there exist $x_K\in A_K$ such that $\mathrm{dist}_X(A_K,A_0) = \mathrm{dist}_X(x_K,A_0)$. Let $K_n\to 0$ be a sequence, we choose its any subsequence $K_\tau$. As $\bigcup_{ K \in (0,K_{max}]} A_K$ is relatively compact, there exists $x\in X$ such that $x_{K_\tau}\to x$, for another subsequence, where by the Kuratowski upper semicontinuous convergence there must hold $x\in A_0$. So,
$$
\mathrm{dist}_X(A_{K_\tau},A_0) = \mathrm{dist}_X(x_{K_\tau},A_0) \leq \rho(x_{K_\tau},x) \to 0.
$$
Hence $\mathrm{dist}_X(A_{K_n},A_0) \to 0$ for the whole sequence $K_n$ and the assertion is proved.  
\end{proof}  
  
\subsection{Result on upper semicontinuous convergence of attractors.} We pass to the proof of \eqref{eq:hausdorff}. First note that every set $\mathcal{A}^K$ for $K\in [0,K_{max}]$, as a global attractor, is compact in $X$.  
Now note, that the bounds in Lemmas \ref{lem:enstrophy_2} and \ref{le:grad-temp-est-H1} are independent of $K$. This means that the absorbing set $B_1\in \mathcal{B}(Z)$ given by Lemma \ref{lem:diss} can be chosen independent of $K$. As $\mathcal{A}^K \subset Y^K \subset B_1$, we deduce that 
\begin{equation}\label{eq:bound}
\bigcup_{K\in [0,K_{max}]} \mathcal{A}_K \subset B_1.
\end{equation}
As $B_1 \in \mathcal{B}(Z)$, and hence this set  is relatively compact in $X$, it follows that $\bigcup_{K \in (0,K_{max}]} \mathcal{A}^K$ is also relatively compact in $X$. So, by Lemma \ref{prop:upper} it suffices to prove the Kuratowski upper-semicontinuous convergence. We will prove the following result.
\begin{theorem}
The global attractors $\mathcal{A}^K$ converge to $\mathcal{A}^0$ upper semicontinuously in Kuratowski sense, and, in consequence \eqref{eq:hausdorff} holds. 
\end{theorem}

\begin{proof}
Let $(u_0^n,\gamma_0^n,\theta_0^n) \in \mathcal{A}^{K_n}$, where $K_n\to 0$ and $(u_0^n,\gamma_0^n,\theta_0^n)\to (u_0,\gamma_0,\theta_0)$ in $X$ as $n\to \infty$. We must show that $(u_0,\gamma_0,\theta_0) \in \mathcal{A}^0$. By Theorem \ref{thm:complete1} there exists the sequence of eternal trajectories $(u^n,\gamma^n,\theta^n):\mathbb{R}\to Y^{K_n}$ such that $(u^n(s),\gamma^n(s),\theta^n(s)) \in \mathcal{A}^{K_n}$ for every $s\in \mathbb{R}$, and $(u^n(0),\gamma^n(0),\theta^n(0))  = (u_0^n,\gamma_0^n,\theta_0^n)$.  
The inclusion \eqref{eq:bound} implies that
$$
\| (u^n(s),\gamma^n(s),\theta^n(s)) \|_{Z} \leq C\quad \textrm{for every}\quad n\in \mathbb{N}\ \textrm{and}\ s\in \mathbb{R}.
$$ 
As in the proof of Lemma \ref{lem:strong} we deduce that
$$
\int_{t_1}^{t_2}\|P\Delta u^n(r)\|_2^2 + \|\Delta \gamma^n(r)\|_2^2+ \|\Delta \theta^n(r)\|_2^2\, dr\leq C(1+t_2-t_1)\quad \textrm{for every}\quad t_1<t_2,  
$$
and
$$
\int_{t_1}^{t_2}\|u_t^n(r)\|_2^2 + \|\gamma_t^n(r)\|_2^2+ \| \theta_t^n(r)\|_2^2\, dr\leq C(1+t_2-t_1)\quad \textrm{for every}\quad t_1<t_2.  
$$
By a standard diagonal argument we deduce that there exists the triple of functions
$$
(u,\gamma,\theta):\mathbb{R}\to Z,  
$$
such that, for a subsequence, still denoted by $n$, there holds
\begin{align*}
&(u^n,\gamma^n,\theta^n) \to (u,\gamma,\theta) \ \textrm{weakly * in} \ L^\infty_{loc}(\mathbb{R};Z) \ \textrm{and weakly in} \ L^2_{loc}(\mathbb{R};D(-P\Delta)\times D_3(-\Delta)\times D_1(-\Delta)),\\
&(u^n_t,\gamma^n_t,\theta^n_t) \to (u_t,\gamma_t,\theta_t) \ \textrm{weakly in} \ L^2_{loc}(\mathbb{R};X),\\
&(u^n(s),\gamma^n(s),\theta^n(s)) \to (u(s),\gamma(s),\theta(s)) \ \textrm{weakly in} \ Z \ \textrm{for every} \ s\in \mathbb{R}.
\end{align*}
Clearly $(u(s),\gamma(s),\theta(s)) = (u_0,\gamma_0,\theta_0)$. To prove that $(u_0,\gamma_0,\theta_0) \in \mathcal{A}^0$ it suffices to prove that  $(u,\gamma,\theta)$ is a complete trajectory of $\{ S^0(t) \}_{t\geq 0}$. Indeed $(u,\gamma,\theta)$ is $Z$ bounded and hence also $X$ bounded, so the assertion follows by Theorem \ref{thm:complete_2}.
As  $(u^n,\gamma^n,\theta^n)$ are strong solutions, for every $t_1\in \mathbb{R}$ and $t_2>t_1$ there holds 
\begin{align}
& \epsilon\int_{t_1}^{t_2}( u^n_t(t),v(t) ) + ((u^n(t)\cdot \nabla) u^n(t),v(t))\, dt + (1+K_n)\int_{t_1}^{t_2}(-P \Delta u^n(t), v(t))\, dt\nonumber\\
& \qquad  = 2K_n \int_{t_1}^{t_2}({\rm rot}\, \gamma^n(t),v(t))\, dt + {\rm Ra} \int_{t_1}^{t_2}(\theta^n(t),v_3(t))\, dt, \label{eq:strong_1_n}\\
& \epsilon M\int_{t_1}^{t_2}( \gamma^n_t(t),\xi(t) ) + (u^n(t)\cdot \nabla\gamma^n(t),\xi(t))\, dt + L\int_{t_1}^{t_2}(-\Delta \gamma^n(t),\xi(t))\, dt \nonumber\\
&\qquad + G\int_{t_1}^{t_2}(-\nabla {\rm div}\,\gamma^n(t), \xi(t))\, dt +
4 K_n \int_{t_1}^{t_2}(\gamma^n(t),\xi(t))\, dt = 2K_n \int_{t_1}^{t_2}({\rm rot}\, u^n(t),\xi(t))\, dt,\label{eq:strong_2_n}\\
&  \int_{t_1}^{t_2}( \theta^n_t(t), \eta(t) )\, dt  +  \int_{t_1}^{t_2}(u^n(t)\cdot \nabla \theta^n(t),\eta(t))\, dt + \int_{t_1}^{t_2}(-\Delta  \theta^n(t), \eta(t))\, dt  = \int_{t_1}^{t_2}(u_3^n(t),\eta(t))\, dt,\label{eq:strong_3_n}
\end{align}
for every triple of test functions $(v,\xi,\eta) \in L^2(t_1,t_2;X)$. 

We pass with $n$ to infinity. Passing to the limit in terms which do not contain $K_n$ is done exactly as in Lemma \ref{lem:strong}. All terms which contain $K_n$ tend to zero. Indeed, for example
\begin{align*}
&\left|K_n \int_{t_1}^{t_2}(-P \Delta u^n(t), v(t))\, dt \right| \leq K_n \|u^n\|_{L^2(t_1,t_2;D(-P\Delta))}\|v\|_{L^2(t_1,t_1;H)} \\
&\qquad \leq K_n C (\sqrt{1+t_2-t_1}) \|v\|_{L^2(t_1,t_1;H)} \to 0.
\end{align*}
After passing to the limit we deduce that
\begin{align}
& \epsilon\int_{t_1}^{t_2}( u_t(t),v(t) ) + ((u(t)\cdot \nabla) u(t),v(t))\, dt + \int_{t_1}^{t_2}(-P \Delta u(t), v(t))\, dt =  {\rm Ra} \int_{t_1}^{t_2}(\theta(t),v_3(t))\, dt, \label{eq:strong_1_n_l}\\
& \epsilon M\int_{t_1}^{t_2}( \gamma_t(t),\xi(t) ) + (u(t)\cdot \nabla\gamma(t),\xi(t))\, dt + L\int_{t_1}^{t_2}(-\Delta \gamma(t),\xi(t))\, dt \nonumber\\
&\qquad + G\int_{t_1}^{t_2}(-\nabla {\rm div}\,\gamma(t), \xi(t))\, dt  = 0,\label{eq:strong_2_n_l}\\
&  \int_{t_1}^{t_2}( \theta_t(t), \eta(t) )\, dt  +  \int_{t_1}^{t_2}(u(t)\cdot \nabla \theta(t),\eta(t))\, dt + \int_{t_1}^{t_2}(-\Delta  \theta(t), \eta(t))\, dt  = \int_{t_1}^{t_2}(u_3(t),\eta(t))\, dt,\label{eq:strong_3_n_l}
\end{align}
for every $t_1\in \mathbb{R}$, every $t_2>t_1$ and every triple of test functions $(v,\xi,\eta) \in L^2(t_1,t_2;X)$. We have proved that $(u,\gamma,\theta)$ is the eternal strong solution of the problem with $K=0$, Lemma \ref{lemma:weak-strong} implies that it is also a weak solution, and hence the trajectory of $\{ S^0(t) \}_{t\geq 0}$. The proof is complete. 
\end{proof} 

\subsection{Relation between $\mathcal{A}^0$ and the attractor for the Newtonian fluid.}
We start this section from a result on properties of $\mathcal{A}^0$. 
\begin{theorem}
	If $(u^0,\gamma^0,\theta^0)$ belong to $\mathcal{A}^0$ then $\gamma^0 = 0$. 
\end{theorem}
\begin{proof}
	Let $(u,\gamma,\theta)$ be the complete trajectory of the problem with $K=0$ such that $(u(0),\gamma(0),\theta(0)) = (u^0,\gamma^0,\theta^0)$. Taking $\xi(t) = \gamma(t)$ in \eqref{eq:strong_2_n_l} we deduce that
	$$
	\frac{\epsilon M}{2}(\|\gamma(t_2)\|^2-\|\gamma(t_1)\|^2) + L\int_{t_1}^{t_2}\|\nabla \gamma(t)\|_2^2 dt \leq 0.
	$$  
	for every $t_1\in \mathbb{R}$ and $t_2> t_1$. The Poincar\'{e} inequality implies that there exists a constant $C>0$ such that
	$$
	\|\gamma(t_2)\|_2^2 + C\int_{t_1}^{t_2}\|\gamma(t)\|_2^2dt \leq \|\gamma(t_1)\|_2^2, 
	$$
	for every $t_1\in \mathbb{R}$ and $t_2> t_1$. Using \cite[Lemma 7.2]{Ball} we deduce that 
	$$
	\|\gamma(t_2)\|_2^2 \leq e^{C(t_1-t_2)}\|\gamma(t_1)\|_2^2
	$$
	for every $t_1\in \mathbb{R}$ and $t_2> t_1$. In particular, for every $t_1<0$ there holds
	$$
	\|\gamma_0\|_2^2 \leq e^{C t_1}\|\gamma(t_1)\|_2^2
	$$
	As $\|\gamma(t_1)\|_2 \leq C$, we can pass with $t_1$ to $-\infty$ to deduce that $\gamma_0 = 0$. The proof is complete.
\end{proof}
We define $\Pi_{(u,\theta)}$ as the projection on variables $u$ and $\theta$. 
For Newtonian fluids governed by equations \eqref{eq:1f_kz}, \eqref{eq:2f_kz}, and \eqref{eq:4f_kz}
 we proceed  as for micropolar ones, defining the weak solutions as the limit of $u$-Galerkin approximative problems analogously to Definition \ref{def-weak-sol-evol}, and the strong solutions analogously to Definition \ref{def-strong} with the equation for $\gamma$ removed from the system. This allows us to define the appropriate $(H\times E_1 , \Pi_{(u,\theta)}Y^0)$-global attractor for multivalued eventual semiflow governed by the weak solutions of the problem for the Newtonian fluid. The semiflow is denoted by $\{ S^{NEWT}(t)\}_{t\geq 0}$, where the mappings $S^{NEWT}(t):H\times E_1 \to \mathcal{P}(H\times E_1)$ associate to the initial data the value of the weak solution at time $t$.    
In the next result we demonstrate the relation between the global attractor for the Rayleigh--B\'{e}nard problem for the Newtonian fluid and the $\Pi_{(u,\theta)}$  projection of the global attractor $\mathcal{A}^0$. 
\begin{theorem}
	Let ${\rm Pr} \geq 4\sqrt{2}c_1{\rm Ra}{\sqrt{A}}$ and let $M, L$ be any nonnegative numbers such that $M\leq 2L$. The projection $\Pi_{(u,\theta)} \mathcal{A}^0$ is the $(H\times E_1 , \Pi_{(u,\theta)}Y^0)$-global attractor for the three-dimensional Rayleigh--B\'{e}nard problem for the Newtonian fluid, which is furthermore the invariant set.
\end{theorem}
\begin{proof}
As $\mathcal{A}^0$ is nonempty and compact in $Y^0$ (and $Z = V\times W_3 \times W_1$), it is clear that $\Pi_{(u,\theta)} \mathcal{A}^0$ is nonempty and compact in $\Pi_{(u,\theta)}Y^0$. Choose any bounded set $B\in \mathcal{B}(H\times E_1)$. Denote the extension of $B$ to $H$ by 
$$
\mathcal{L}B = \{  (u,0,\theta)\,:\ (u,\theta) \in B \}.
$$ 
This is a bounded set in $H$. Energy inequality \eqref{ine-omega:4f_w} implies that
$$
S^0(t)(u,0,\theta) = \{ (u(t),0,\theta(t))\,:\ (u(t),\theta(t)) \in  S^{NEWT}(t)(u,\theta) \}.
$$ 
Hence 
$$
\mathrm{dist}_{\Pi_{(u,\theta)}Y^0} (S^{NEWT}(t)B,\Pi_{(u,\theta)} \mathcal{A}^0) = \mathrm{dist}_{Y^0} (S^{0}(t)B, \mathcal{A}^0) \to 0\quad \textrm{as}\quad t\to \infty.
$$ 
Invariance of $\mathcal{A}^0$ through $\{ S^0(t) \}_{t\geq 0}$  implies the invariance of $\Pi_{(u,\theta)} \mathcal{A}^0$ through $\{ S^{NEWT}(t)\}_{t\geq 0}$. It is also immediate to see that $\Pi_{(u,\theta)} \mathcal{A}^0$ is the smallest closed attracting set. Indeed, if $C$ is closed and attracting, then it must hold that
$$
0 = \lim_{t\to \infty}\mathrm{dist}_{\Pi_{(u,\theta)}Y^0} (S^{NEWT}(t)\Pi_{(u,\theta)} \mathcal{A}^0,C) = \lim_{t\to \infty}\mathrm{dist}_{\Pi_{(u,\theta)}Y^0} (\Pi_{(u,\theta)} \mathcal{A}^0,C).
$$ 
Hence 
$$
\Pi_{(u,\theta)} \mathcal{A}^0 \subset C,
$$
and the proof is complete.
\end{proof}




\begin{bibdiv} 
\begin{biblist}
%







	\bib{Babin_Vishik}{book}{
	author={Babin, A.V.},
	author={Vishik, M.I.},
	title={Attractors of Evolution Equations},
	publisher={North Holland},
	place={Amsterdam, London, New York, Tokyo},
	date={1992},
} 

\bib{Ball}{article}{
	author={Ball, J.M.},
	title={Continuity Properties and Global Attractors of Generalized Semiflows and the Navier-Stokes Equations},
	journal={Journal of Nonlinear Science},
	volume={7},
	date={1997},
	pages={475--502},
}


%

\bib{CLR2013}{book}{
	author={Carvalho, A.N.},
	author={Langa, J.A.},
	author={Robinson, J.C.},
	title={Attractors for infinite-dimensional
		non-autonomous dynamical systems},
	date={2013},
	publisher={Springer},
}

\bib{Chepyzhov_Conti_Pata}{article}{
author={Chepyzhov, V.V.},
author={Conti, M.},
author={Pata, V.},
title={A minimal approach to the theory of global attractors},
journal={Discrete and Continuous Dynamical Systems},
volume={32},
date={2012},
pages={2079--2088},
} 

%
\bib{CD1}{article}{
	author={Constantin, P.},
   author={Doering, C.R.},
	title={Heat transfer in convective turbulence},
	journal={Nonlinearity},
	volume={9},
	date={1996},
	pages={1049--1060},
}


\bib{CD2}{article}{
	author={Constantin, P.},
   author={Doering, C.R.},
	title={Variational bounds on energy dissipation in incompressible flows. III. Convection},
	journal={Physical Review E},
	volume={53},
	date={1996},
	pages={5957--5981},
}
%

%
%

\bib{DMP}{book}{
	author={Denkowski, Z.},
	author={Mig\'{o}rski, S.},
	author={Papageorgiou, N.S.},
	title={An Introduction to Nonlinear Analysis: Theory},
	date={2003},
	publisher={Kluwer Academic Publishers},
}

\bib{DG}{book}{
	author={Doering, C.R.},
	author={Gibbon, J.D.},
	title={Applied Analysis of the Navier-Stokes equations},
	date={1995},
	publisher={Cambridge Texts in Applied Mathematics},
}


\bib{Eringen66}{article}{
	author={Eringen, A.C.},
	title={Theory of micropolar fluids},
	journal={J. Math. Mech.},
	volume={16},
	date={1966},
	pages={1--16},
}

%

\bib{Foias_MRT-NSE_Turb}{book}{
   author={Foias, C.},
   author={Manley, O.},
   author={Rosa, R,},
   author={Temam, R.},
   title={Navier--Stokes Equations and Turbulence}, 
   date={2001},
   publisher={CUP},
   
}   

%
%
%
%
%
%
%
%
%
%
%

\bib{KaLaLu_2017_MMN}{article}{
    author={Kalita, P.},
    author={Langa, J.},
	author={\L{}ukaszewicz, G.},
	title={Micropolar meets Newtonian. The Rayleigh--B\'{e}nard problem},
	journal={Physica D: Nonlinear Phenomena},
	doi={10.1016/j.physd.2018.12.004},
	date={2019}
}


\bib{KLS2}{article}{
	author={Kalita, P.},
	author={\L{}ukaszewicz, G.},
	author={Siemianowski, J.},
	title={On relation between attractors for single and multivalued semiflows for a certain class of PDEs},
	journal={Discrete and Continuous Dynamical Systems B},
	date={2019},
	volume={24},
	pages={1199--1227},
}




\bib{L1999}{book}{
   author={\L{}ukaszewicz, G.},
   title={Micropolar Fluids - Theory and Applications},
   date={1999},
   publisher={Birkh\"{a}user Basel},
}

\bib{L2001}{article}{
   author={\L{}ukaszewicz, G.},
   title={Long time behavior of 2D micropolar fluid flows},
   journal={Mathemathical and Computer Modelling},
   volume={34},
   date={2001},
   pages={487--509},

}



%


%




	\bib{Melnik_Valero}{article}{
	author={Melnik, V.S.},
	author={Valero, J.},
	title={On attractors of multivalued semiflows and differential inclusions},
	journal={Set-Valued Anal.},
	volume={6},
	date={1998},
	pages={83--111},      
}  

\bib{Melnik-2008}{article}{
	author={Melnik, V.S.},
	author={Valero, J.},
	title={Addendum to ''On Attractors of Multivalued
		Semiflows and Differential Inclusions'' [Set-Valued Anal. 6 (1998), 83-111],},
	journal={Set-Valued Anal.},
	volume={16},
	date={2008},
	pages={507--509},      
}



\bib{PayneStraughan}{article}{
	author={Payne, L.E.},
	author={Straughan, B.},
	title={Critical Rayleigh numbers for oscillatory and nonlinear convection in an isotropic thermomicropolar fluid},
	journal={International Journal of Engineering Science},
	volume={27},
	date={1989},
	pages={827--836},
	
}


%


\bib{Robinson}{book}{
	author={Robinson, J.C.},
	title={Infinite-Dimensional Dynamical Systems},
	publisher={Cambridge University Press},
	place={Cambridge, UK},
	date={2001},
} 

\bib{Rummler_2}{article}{
 author={Rummler, B.},
 title={The eigenfunctions of the {S}tokes operator in special domains. {II}},
 journal={Z. Angew. Math. Mech.},
 volume={17},
 date={1980},
 pages={669--675},
}


%


%
\bib{Tarasinska_2006_paper}{article}{
	author={Tarasi\'{n}ska, A.},
	title={Global attractor for heat convection problem in a micropolar fluid},
	date={2006},
	journal={Math. Meth. Appl. Sci.},
	volume={29},
	pages={1215--1236},
}
%

\bib{Temam}{book}{
	author={Temam, R.},
	title={Infinite Dimensional Dynamical Systems in Mechanics and Physics},
	date={1997},
	publisher={Springer},
}

\bib{Wang_chapter}{collection}{
	author={Wang, X.},
	title={A Note on Long Time Behavior of Solutions to the Boussinesq System at Large Prandtl Number},
	editor={Gui-Qiang Chen, George Gasper, Joseph W. Jerome},
	booktitle={Nonlinear Partial Differential Equations And Related Analysis},
	publisher   = {American Mathematical Society},
	volume={371},
	series={Contemporary Mathematics},
	address     = {Providence, Rhode Island},
	year        = {2005},
	pages       = {315--323},
}


\bib{Wang2007-Asymptotic}{article}{
	author={Wang, X.},
	title={Asymptotic Behavior of the Global Attractors to the Rayleigh--B\'{e}nard Convection at Large Prandtl Number},
	journal={Commun. Pure Appl. Math.},
	volume={LX},
	date={2007},
	pages={1293--1318},
	
}


\bib{Wang2008}{article}{
	author={Wang, Xiaoming},
	title={Bound on vertical heat transport at large Prandtl number},
	journal={Physica D},
	volume={237},
	date={2008},
	pages={854--858},
	
}




\end{biblist}
\end{bibdiv}
\end{document}